\documentclass[onecolumn]{IEEEtran}

\usepackage{macro}

\title{List Decoding for Oblivious Arbitrarily Varying MACs: Constrained and Gaussian}

\author{
\IEEEauthorblockN{
Yihan Zhang\IEEEauthorrefmark{1}
}\\
\IEEEauthorblockA{
\IEEEauthorrefmark{1}Dept.\ of Information Engineering, The Chinese University of Hong Kong \\\href{mailto:zy417@ie.cuhk.edu.hk}{zy417@ie.cuhk.edu.hk} 
}
}

\begin{document}
\maketitle

\begin{abstract}
This paper provides upper and lower bounds on list sizes of list decoding for two-user oblivious arbitrarily varying multiple access channels (AVMACs). 
An oblivious AVMAC consists of two users who wish to transmit messages  (without cooperation) to a remote receiver, a malicious jammer who  only has access to the codebooks of both users (which are also known to every party), and a receiver who is required to decode the message pair sent by both users.
The transmitters send codewords which encode messages subject to input constraints.
The jammer, without knowing the transmitted codeword pair, injects adversarial noise subject to state constraints so as to actively corrupt the communication from both users to the receiver. 
It was left as an open question in \cite{cai-2016-list-dec-obli-avmac} to nail down the smallest list sizes for \emph{constrained} AVMACs.
Our inner and outer bounds are based on a judicious notion of symmetrizability for AVMACs introduced by \cite{cai-2016-list-dec-obli-avmac} with twists to incorporate input and state constraints. 
The analysis follows techniques by Csisz\'ar and Narayan \cite{csiszar-narayan-it1988-obliviousavc}.
When no constraints are imposed, our bound collapse to prior results by Cai \cite{cai-2016-list-dec-obli-avmac} which characterized the list-decoding capacity region of \emph{unconstrained} AVMACs. 
Techniques used in this paper can also be extended to the Gaussian case and we \emph{characterize} the list-decoding capacity region for Gaussian AVMACs. 
The converse argument  relies on a  bounding technique recently used by Hosseinigoki and Kosut \cite{hosseinigoki-kosut-2018-oblivious-gaussian-avc-ld}.
\end{abstract}

\section{Introduction}\label{sec:intro}
\emph{Oblivious arbitrarily varying channels} (AVCs), introduced by Blackwell, Breiman and Thomasian \cite{blackwell-avc-1960}, models communication media that is governed by active adversaries with limited knowledge. 
Specifically,  AVCs are channels which takes transmitted signals as inputs and outputs signals according to the \emph{state} of the channel which may vary in an arbitrary manner as the adversary desires. 
The goal of the adversary, who we call James\footnote{He is named so since he can maliciously ``jam'' the channel.}, is to prevent communication from happening from the input end to the output end by introducing carefully designed (not necessarily randomly drawn from certain fixed distribution) noise. 
It turns out that the knowledge that James possesses plays an crucial role in the study of AVCs. 
We say that the  channel (or the adversary) is \emph{oblivious} if the adversary only has access to the codebook(s) used by the transmitter(s), but not the actually transmitted signals. 
Put in other words, the noise James injects cannot depend on the transmitted codeword; or, he is required to fix his jamming vector before the the transmission is instantiated.
On the contrary, if James does not only know the codebook but also the transmitted codeword, then he is said to be \emph{omniscient} \cite{csiszar-korner-1981}.
The study of omniscient AVCs  essentially boils down to zero-error combinatorial questions regarding high-dimensional packing and the capacity for such channels are widely open even for very simple AVCs, e.g., bit-flip channels.
Oblivious AVCs serve as an interpolation between the worst-case model, omniscient AVCs, and the average-case model, Shannon channels, i.e., channels with random noise obeying certain fixed distribution. 
In the  point-to-point scenario, there have been a handful of capacity results.
Empirically, the capacity of point-to-point oblivious AVCs exhibits similar behaviours to the capacity of its Shannon counterpart.
Indeed,  it is provably known \cite{csiszar-narayan-it1988-obliviousavc} that the best strategy for James is essentially to mimic a Shannon channel, i.e., transmitting random noise.

In terms of model, this paper is a continuation of this line of research towards multiuser setting, in particular, the two-user multiple access setting. 
Informally, (two-user) AVMACs model uplink communication with an oblivious adversary. 
Two transmitters who are not allowed to cooperate both want to send messages to a single receiver.
The channel takes two codewords from both users and transforms it according to the channel transition law. 
James gets to control the channel law by choosing a state sequence only based on two users' codebooks (which are public to every party).
The channel follows a different law for each different state.
The receiver, receiving a noisy word output by the channel, aims to estimate both messages reliably.

In terms of communication goal, this paper pushes our understanding beyond unique decoding capacity. 
Instead of insisting on the decoder to exactly reconstruct the transmitted message, we relaxed the goal and allow the decoder to output a \emph{list} of messages required to contain the correct message. 
Such a requirement is known as \emph{list decoding}, introduced by Elias \cite{elias-1957-listdec} and Wozencraft \cite{wozencraft-1958-listdec}. 
It was extensively studied against both worst-case and average-case errors. 
For worst-case notion of list decoding, improving the performance  and constructing explicit list-decodable codes attracted much attention in computer science community.
Despite being interesting in its own right, the concepts and techniques of worst-case list decoding finds numerous applications  in computational complexity \cite{guruswami-2006-list_dec_avgcasecplx}, the theory of pseudorandomness \cite{dmoz-2019-pseudorandom_from_hardness}, learnings theory \cite{dks-2018-list_dec_est_and_learn}, cryptography \cite{goldreich-levin-1989-hardcorepredicate-listdec}, etc. 
As for list decoding for non-omniscient channels, besides being an important subject by itself, list decoding is a useful primitive which allows us to invoke as a proof technique to get intermediate results \cite{chen-et-al}. 
In many cases, it turns out that one can first list decode to a small sized uncertainty set and then disambiguate it using extra information. 

\section{Problem formulation}\label{sec:problem_formulation}
Throughout this paper, consider an oblivious arbitrarily varying multiple access channel (AVMAC) 
\[\cA = (\cX,\cY,\cS,\cZ,f_1,f_2,\Gamma_1,\Gamma_2,g,\Lambda,W_{\bfz|\bfx,\bfy,\bfs}),\]
formally as follows.
The message sets of transmitter one and two are denoted by $ \cM\coloneqq[ L2^{nR_1}] $ and $ \cW\coloneqq[ L2^{nR_2}] $, respectively.
The messages $ \bfm $ and $ \bfw $ to be transmitted by user one and two are assumed to be uniformly distributed in $ \cM $ and $ \cW $, respectively.
For any $ m\in[ L2^{nR_1}],w\in[ L2^{nR_2}] $,
encoder one and two encode them into $\vx_1\in\cX^n$ and $\vy\in\cY^n$ respectively. The adversary designs an adversarial noise $\vs\in\cS^n$ \emph{only} based on his knowledge of the codebooks used by both encoders (not based on any knowledge of the transmitted codewords). Given the  channel output $\vz$, the receiver aims to decode to a list $\cL$ of at most $L$ message pairs which contains the transmitted $(m,w)$. We impose input and state constraints as follows. Let
\[f_1:\cX\to\bR_{\ge0},\;f_2\colon \cY\to\bR_{\ge0}\]
be cost functions of input symbols and let $g:\cS\to\bR_{\ge0}$ be a cost function of state symbols. Further define
\begin{align*}
    f_1(\vx) = &\frac{1}{n}\sum_if_1(\vx(i)),\;f_2(\vy)=\frac{1}{n}\sum_if_2(\vy(i)),\\
    g(\vs) = &\frac{1}{n}\sum_ig(\vs(i)).
\end{align*}
We require all codewords to satisfy
\[f_1(\vx)\le\Gamma_1,\;f_2(\vy)\le\Gamma_2,\]
and every state vector to satisfy $g(\vs)\le\Lambda$. See Fig. \ref{fig:list_dec_obli-avmac} for the system diagram of list decoding for oblivious AVMACs.
\begin{figure}[htbp]
  \centering
  \includegraphics[width=0.95\textwidth]{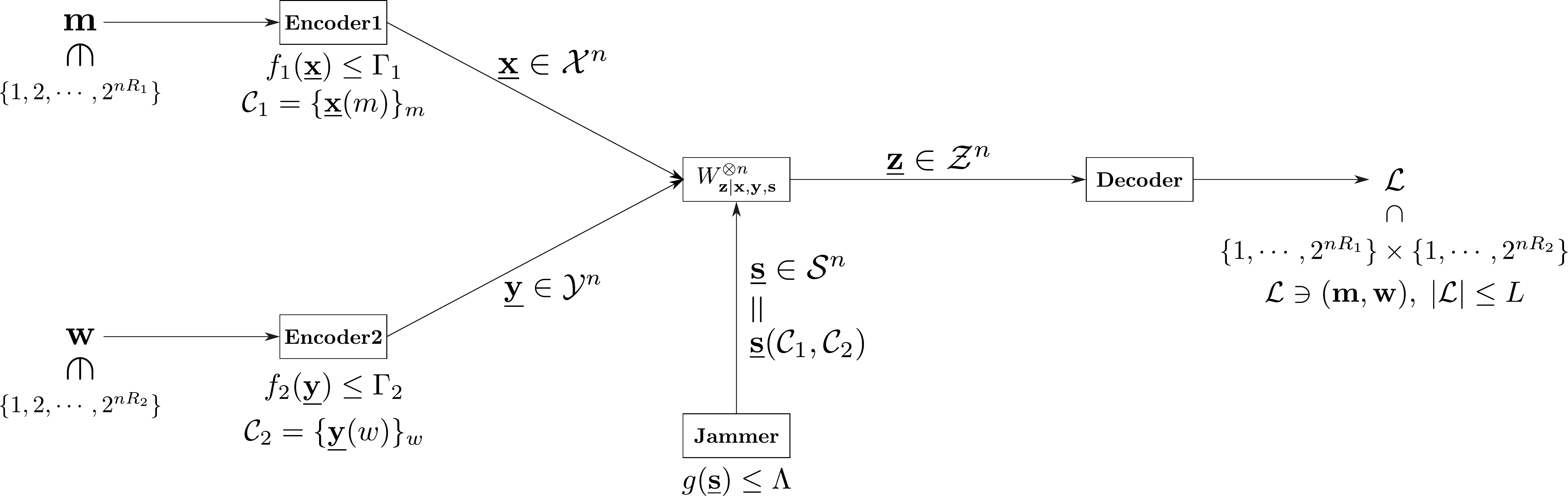}
  \caption{List decoding for oblivious AVMACs.}
  \label{fig:list_dec_obli-avmac}
\end{figure}

We are interested in proving inner and outer bounds on the $L$-list decoding capacity region for oblivious AVMACs described above. 

\section{Organization of the paper}\label{sec:organization}
The rest of the paper is structured as follows.
In Sec. \ref{sec:prior_work}, we survey relevant prior work pertaining oblivious single/multi-user AVCs with discrete/continuous alphabet, in the unique/list decoding setting.
Our main results regarding list-decoding capacity of input-and-state-constrained oblivious AVMACs are stated in Sec. \ref{sec:results}.
Before proceeding with the full proof, we fix our notational convention in Sec. \ref{sec:notation} and provide necessary preliminaries in Sec. \ref{sec:prelim}.
Inner and outer bounds in our main theorem are proved in Sec. \ref{sec:obli_avmac_constrained}.
Analogous results and their proofs for Gaussian channels are stated and sketched in Sec. \ref{sec:gaussian_avmac}.

\section{Prior work}\label{sec:prior_work}
We are only concerned with deterministic code capacity of oblivious adversarial channels. 
\subsection{Discrete alphabet}

The capacity of oblivious AVCs with and without constraints was given by \cite{csiszar-narayan-it1988-obliviousavc}. Hughes \cite{hughes-1997-list-avc} used their techniques to obtain the $L$-list decoding capacity of oblivious AVCs without constraint for any $L$. When state constraints are imposed, upper and lower bounds on $L$-list decoding capacity of oblivious AVCs were given by \cite{sarwate-gastpar-2012-list-dec-avc-state-constr}. They do not match in general for reasons we illustrate later. The $L$-list decoding capacity of the Gaussian counterpart is obtained by \cite{hosseinigoki-kosut-2018-oblivious-gaussian-avc-ld}. The work \cite{sarwate-gastpar-2012-list-dec-avc-state-constr} and \cite{hosseinigoki-kosut-2018-oblivious-gaussian-avc-ld} used essentially  the same techniques as \cite{csiszar-narayan-it1988-obliviousavc}. 

For oblivious AVMACs, the capacity region is given by \cite{ahlswedecai-1999-obli-avmac-no-constr}. However, their result only gave positive rate when the capacity region has nonempty interior. The characterization is obtained when their result is combined with the dichotomy theorem. Their techniques (so-called elimination techniques) do not work in the presence of state constraints. Recently, \cite{pereg-steinberg-2019-obli-avmac-with-wo-constr} gave a characterization of the capacity region of oblivious AVMACs with (and without) constraints using \cite{csiszar-narayan-it1988-obliviousavc}'s techniques. This, in particular, recovers the result by \cite{ahlswedecai-1999-obli-avmac-no-constr} without resorting to the dichotomy theorem.

For list decoding oblivious AVMACs, Cai \cite{cai-2016-list-dec-obli-avmac}  proposed a judicious notion of symmetrizability and used the elimination technique to obtain the $L$-list decoding capacity of oblivious AVMACs without state constraints. 

Apparently, two missing pieces along this line of research is  the $L$-list decoding capacity region of AVMACs \emph{with} constraints and that of  the Gaussian counterpart. 

\subsection{Continuous alphabet}
For point-to-point single-user oblivious Gaussian AVCs, the deterministic code capacity   is determined by Csisz\`ar and Narayan 
\cite{csiszar-narayan-1991-gavc}. The list decoding capacity is recently given by \cite{hosseinigoki-kosut-2018-oblivious-gaussian-avc-ld}. For two-user oblivious Gaussian AVMACs, the deterministic code capacity region is obtained as a corollary in \cite{pereg-steinberg-2019-obli-avmac-with-wo-constr}. In what follows, we aim to nail down the list decoding capacity region of oblivious Gaussian AVMACs.

\section{Main results}\label{sec:results}
We state inner and outer bounds that we are going to prove in the rest of this paper. To this end, we need a sequence of definitions. 

Define the collection of (generic) bipartite graphs:
\begin{align}
    \sB_L\coloneqq&\curbrkt{\cB=(\cI,\cJ,\cE)\colon \begin{array}{l}
         \cI = \curbrkt{1,\cdots,I},\; I\le L,\\
         \cJ = \curbrkt{1,\cdots,J},\; J\le L,\\
         \cE\subset\cI\times\cJ,\;|\cE| \le L  
    \end{array} }.\notag
\end{align}
We assume that vertices in $\cI,\cJ$ are listed in ascending order. Let 
\begin{align}
    \sL_L(m,w)\coloneqq&\curbrkt{\cL = (\cS,\cT,\cF)\colon \begin{array}{l}
         \cS\subset[M],\cT\subset[W],  \\
         \cF\subset\cS\times\cT,\;\card{\cF}\le L, \\
         m\in\cS,w\in\cT,(m,w)\in\cF
    \end{array} }\notag
\end{align}
denote the collection of bipartite graphs realized by messages. We assume that messages in $\cS,\cT$ are listed in ascending order.

Let $\cB = (\cI,\cJ,\cE)$ denote a bipartite graph with left vertex set $\cI = [I]$, right vertex set $\cJ = [J]$ and edge set $\cE \subset \cI\times\cJ$.
Assume $\cB$ has no isolated vertex. 

For a bipartite graph $\cB = (\cI,\cJ,\cE)$, define the set of $\cB$-symmetrizing distributions 
\begin{align*}
    \cQ_\symm(\cB) \coloneqq &\curbrkt{ Q^{(\cB)}_{\bfs|\bfx^{I-1},\bfy^{J-1}}\in\Delta(\cS|\cX^{I-1}\times\cY^{J-1})\colon  \begin{array}{l}
          \sum_{s}W(z|x_i,y_j,s)Q(s|x^{\cI\setminus i},y^{\cJ\setminus j}) \\
          = \sum_s W(z|x_{\sigma(i)},y_{\pi(j)})Q(s|x^{\sigma(\cI\setminus i)},y^{\pi(\cJ\setminus j)}),\\
          \text{for all }(i,j)\in\cE,\\
          \text{for all }\sigma\in S_I,\pi\in S_J\text{ s.t. }(\sigma,\pi)\cE = \cE,\\
          \text{for all }x^\cI\in\cX^I,y^\cJ\in\cY^J,z\in\cZ
    \end{array} }.
\end{align*}
Note that both sides of the equation in the definition is a distribution on $\cZ|\cX^I\times\cY^J$. This is a judicious notion due to Cai \cite{cai-2016-list-dec-obli-avmac}.

We now define two notions of symmetrizability. Symmetrizability is the largest list size the adversary can cause subject to his constraints. For technical reasons that we will illustrate later, we need a strong notion for outer bound and a weak notion for inner bound.

Define the strong symmetrizability $L_{s}(P_{\bfu,\bfx,\bfy})$ w.r.t. $P_{\bfu,\bfx,\bfy}$ as
\begin{align*}
    &L_s(P_{\bfu,\bfx,\bfy})\\
    \coloneqq &\max\curbrkt{ L\colon  \begin{array}{l}
         \exists \cB = (\cI,\cJ,\cE)\text{ s.t. } |\cE| = L, \\
         \displaystyle\max_{\substack{ P_{\bfu,\bfx^{I-1},\bfy^{J-1}}\\P_{\bfu,\bfx_i,\bfy_j} = P_{\bfu,\bfx,\bfy},\;\forall (i,j)\in\cE}}\min_{ \curbrkt{Q^{(u)}_{\bfs|\bfx^{I-1},\bfy^{J-1}}}_u\subset\cQ_\symm(\cB) }  \sum_{u,x^{I-1},y^{J-1},s}P(u)P(x^{I-1},y^{J-1}|u)Q^{(u)}(s|x^{I-1},y^{J-1})g(s)<\Lambda
    \end{array} }.
\end{align*}
Note that the max and min can be reversed since the objective function is linear. Strong symmetrizability will be used to obtain outer bounds. 

Define the weak symmetrizability $L_{w}(P_{\bfu,\bfx,\bfy})$ w.r.t. $P_{\bfu,\bfx,\bfy}$ as
\begin{align*}
    &L_w(P_{\bfu,\bfx,\bfy})\\
    \coloneqq &\max\curbrkt{ L \colon  \begin{array}{l}
         \exists\cB = (\cI,\cJ,\cE)\text{ s.t. } |\cE| = L, \\
         \displaystyle\min_{\curbrkt{ Q^{(u)}_{\bfs|\bfx^{I-1},\bfy^{J-1}}}_u\subset\cQ_\symm(\cB) }  \sum_{u,x^{I-1},y^{J-1},s}P(u)P_{\bfx|\bfu}^{\otimes(I-1)}(x^{I-1}|u)P_{\bfy|\bfu}^{\otimes(J-1)}(y^{J-1}|u)Q^{(u)}(s|x^{I-1},y^{J-1})g(s)<\Lambda
    \end{array} }.
\end{align*}
Weak symmetrizability will be used to obtain inner bounds. The difference from the strong one is that the maximization over all joint distributions of $\bfx^{I-1},\bfy^{J-1}$ is replaced by the product distribution. It is easy to see that $L_s(P_{\bfu,\bfx,\bfy})\le L_w(P_{\bfu,\bfx,\bfy})$ for any $ P_{\bfu,\bfx,\bfy} $.

We are ready to state our inner and outer bounds on the $L$-list decoding capacity region of $\cA$. Fix $L$, the inner bound reads that
\begin{align*}
    C_{\inner} = &\bigcup_{\substack{P_{\bfu,\bfx,\bfy} = P_{\bfu}P_{\bfx|\bfu}P_{\bfy|\bfu}\\ \expt{f_1(\bfx)}\le\Gamma_1,\;\expt{f_2(\bfy)}\le\Gamma_2\\ L_w(P_{\bfu,\bfx,\bfy})<L}}\curbrkt{(R_1,R_2)\colon \begin{array}{rl}
        R_1\le&\inf I(\bfx;\bfz|\bfy,\bfu), \\
        R_2\le&\inf I(\bfy;\bfz|\bfx,\bfu), \\
        R_1+R_2\le&\inf I(\bfx,\bfy;\bfz|\bfu)
    \end{array}  }.
\end{align*}
Both expectations are taken over $ P_{\bfu,\bfx,\bfy} $.
All infimums are taken over jamming distributions $P_{\bfs|\bfu}\in\Delta(\cS|\cU)$ such that
$\sum_{u,s}P(u)P(s|u)g(s)\le\Lambda$. All mutual information is evaluated w.r.t. the distribution $P_{\bfu}P_{\bfx|\bfu}P_{\bfy|\bfu}P_{\bfs|\bfu} W_{\bfz|\bfx,\bfy,\bfs}$.

Replacing the weak symmetrizability $L_w(P_{\bfu,\bfx,\bfy})$ in $C_{\inner}$ with $L_s(P_{\bfu,\bfx,\bfy})$, we get our outer bound $C_{\out}$,
\begin{align*}
    C_{\out} = &\bigcup_{\substack{P_{\bfu,\bfx,\bfy} = P_{\bfu}P_{\bfx|\bfu}P_{\bfy|\bfu}\\ \expt{f_1(\bfx)}\le\Gamma_1,\;\expt{f_2(\bfy)}\le\Gamma_2\\ L_s(P_{\bfu,\bfx,\bfy})<L}}\curbrkt{(R_1,R_2)\colon \begin{array}{rl}
        R_1\le&\inf I(\bfx;\bfz|\bfy,\bfu), \\
        R_2\le&\inf I(\bfy;\bfz|\bfx,\bfu), \\
        R_1+R_2\le&\inf I(\bfx,\bfy;\bfz|\bfu)
    \end{array}  }.
\end{align*}

Define 
\begin{align}
L_s^* \coloneqq& \min\curbrkt{ L_s(P_{\bfu,\bfx,\bfy})\colon \begin{array}{rl}
P_{\bfu,\bfx,\bfy} =& P_{\bfu} P_{\bfx|\bfu} P_{\bfy|\bfu}, \\
\expt{f_1(\bfx)}\le& \Gamma_1, \\
\expt{f_2(\bfy)}\le& \Gamma_2
\end{array} }, \notag \\
L_w^* \coloneqq& \min\curbrkt{ L_w(P_{\bfu,\bfx,\bfy})\colon \begin{array}{rl}
P_{\bfu,\bfx,\bfy} =& P_{\bfu} P_{\bfx|\bfu} P_{\bfy|\bfu}, \\
\expt{f_1(\bfx)}\le& \Gamma_1, \\
\expt{f_2(\bfy)}\le& \Gamma_2
\end{array} }. \notag 
\end{align}
If  $ L\le L_s^* $, then the $L$-list decoding capacity region is $  \curbrkt{ (0,0) } $. 

It is well known that by letting James transmit random noise drawn from certain i.i.d. distribution, the outer bound follows from the strong converse to list decoding (non-adversarial) MACs. 
Hence we omit the proof.

\section{Notation}\label{sec:notation}
\noindent\textbf{Random variables, vectors and matrices.}
Random variables are denoted by lower case letters in boldface or capital letters in plain typeface, e.g., $\bfm,\bfx,\bfs,U,W$, etc. Their realizations are denoted by corresponding lower case letters in plain typeface, e.g., $m,x,s,u,w$, etc. Vectors (random or fixed) of length $n$, where $n$ is the blocklength without further specification, are denoted by lower case letters with  underlines, e.g., $\vbfx,\vbfs,\vx,\vs$, etc. The $i$-th entry of a vector $\vx\in\cX^n$ is denoted by $\vx(i)$  since we can alternatively think $\vx$ as a function from $[n]$ to $\cX$. Same for a random vector $\vbfx$. 
Alternatively, we use $x^k$ to denote a length-$k$ $\cX$-valued vector $ x^k \coloneqq (x_1,\cdots,x_k) $.
For a finite index set $ \cI\subset\bZ_{>0} $, we use $ x^\cI $ to denote an $ \cX $-valued vector of length-$|\cI|$, each component of which is labelled by the corresponding element in $ \cI $.
For example, if $ \cI = \curbrkt{2,3,5,6,9} $, then $ x^{\cI} = (x_2,x_3,x_5,x_6,x_9) $.
Note that $ x^k = x^{[k]} $ in our convention, though we do not pursue the latter notation in this case.
Matrices are denoted by capital letters in boldface, e.g., $\bfP,\mathbf{\Sigma}$, etc. 
We sometimes write $\bfG_{n\times m}$ to explicitly specify its dimension. For square matrices, we write $\bfG_n$ for short. Letter $\bfI$ is reserved for identity matrix.  

\noindent\textbf{Sets.}
For $M\in\bZ_{>0}$, we let $[M]$ denote the set of first $M$ positive integers $\{1,2,\cdots,M\}$.
Sets are denoted by capital letters in calligraphic typeface, e.g., $\cC,\cI$, etc. 
With slight abuse of notation, a singleton set $\{a\}$ is still denoted by $a$. 
The same convention is followed when set operations are performed, e.g., $ \cA\setminus a = \cA\setminus\{a\}, a\cup b = \{a\}\cup\{b\} = \{a,b\} $, etc.
For any finite set $\cX$ and any integer $0\le k\le |\cX|$, we use $\binom{\cX}{k}$ to denote the collection of all subsets of $\cX$ of size $k$, i.e., 
\[\binom{\cX}{k}\coloneqq\curbrkt{\cY\subseteq\cX\colon\card{\cY}=k}.\]
Similarly, let 
\begin{align}
\binom{\cX}{\le k}\coloneqq\curbrkt{\cY\subseteq\cX\colon|\cY|\le k} \notag
\end{align}
denote the collection of all subsets of $\cX$ of size at most $k$.

An $n$-dimensional Euclidean ball centered at $\vx$ of radius $r$ is denoted by
\[\cB^n(\vx,r) \coloneqq \curbrkt{\vy\in\bR^n\colon \normtwo{\vy} \le r}.\]

We use $S_\cA$ or $S_{|\cA|}$ to denote the symmetric group on a finite set $\cA$. 
Permutations are typically denoted by lower case Greek letters.

\noindent\textbf{Functions.}
We use the standard Bachmann--Landau (Big-Oh) notation for asymptotics of real-valued functions in positive integers. 
Throughout the whole paper, $\log$ is to the base 2.
For $x\in\bR$, let $[x]^+\coloneqq\max\curbrkt{x,0}$.
For any $\cA\subseteq\Omega$, the indicator function of $\cA$ is defined as, for any   $x\in\Omega$, 
\[\one_{\cA}(x)\coloneqq\begin{cases}1,&x\in \cA\\0,&x\notin \cA\end{cases}.\]
At times, we will slightly abuse notation by saying that $\one_{\sfA}$ is $1$ when event $\sfA$ happens and 0 otherwise. Note that $\one_{\cA}(\cdot)=\indicator{\cdot\in\cA}$.
Let $ \normtwo{\cdot} $ denote the Euclidean/$L^2$-norm. Specifically, for any $\vx\in\bR^n$,
\[ \normtwo{\vx} \coloneqq\paren{\sum_{i=1}^n\vx_i^2}^{1/2}.\]

\noindent\textbf{Probability.}
The probability mass function (p.m.f.) of a discrete random variable $\bfx$ or a random vector $\vbfx$ is denoted by $P_{\bfx}$ or $P_{\vbfx}$, i.e., 
\[P_\bfx(x) \coloneqq \probover{\bfx\sim P_\bfx}{\bfx = x},\quad P_\vbfx(\vx) = \probover{\vbfx\sim P_\vbfx}{\vbfx = \vx},\]
for any $x\in\cX$ or $\vx\in\cX^n$.
If every entry of $\vbfx$ is independently and identically distributed (i.i.d.) according to $P_{\bfx}$, then we write $\vbfx\sim P_{\bfx}^{\tn}$, where $P_\bfx^\tn$ is a product distribution defined as
\[P_{\vbfx}(\vx)=P_{\bfx}^{\tn}(\vx)\coloneqq\prod_{i=1}^nP_{\bfx}(\vx(i)).\]
For a finite set $\cX$, $\Delta(\cX)$ denotes the probability simplex on $\cX$, i.e., the set of all probability distributions supported on $\cX$,
\[\Delta(\cX)\coloneqq\curbrkt{P_\bfx\in[0,1]^{\card{\cX}}\colon\sum_{x\in\cX}P_\bfx(x) = 1 }.\]
Similarly, $\Delta\paren{\cX\times\cY}$ denotes the probability simplex on $\cX\times\cY$,
\[\Delta\paren{\cX\times\cY}\coloneqq\curbrkt{P_{\bfx,\bfy}\in[0,1]^{\cardX\times\cardY}\colon\sum_{x\in\cX}\sum_{y\in\cY} P_{\bfx,\bfy}(x,y) = 1 }.\]
Let $\Delta(\cY|\cX)$ denote the set of all conditional distributions,
\[\Delta(\cY|\cX)\coloneqq\curbrkt{P_{\bfy|\bfx}\in\bR^{\cardX\times\cardY}\colon P_{\bfy|\bfx}(\cdot|x)\in\Delta(\cY),\;\forall x\in\cX}.\]
The general notion for multiple spaces is defined in the same manner.
For a joint distribution $P_{\bfx,\bfy}\in\Delta(\cX\times\cY)$, let $\sqrbrkt{P_{\bfx,\bfy}}_\bfx\in\Delta(\cX)$ denote the \emph{marginalization} onto the  variable $\bfx$, i.e., for $x\in\cX$,
\[\sqrbrkt{P_{\bfx,\bfy}}_\bfx(x) \coloneqq \sum_{y\in\cY}P_{\bfx,\bfy}(x,y).\]
Sometimes we simply write it as $P_\bfx$ (induced by $ P_{\bfx,\bfy} $) when the notation is not overloaded.


\section{Preliminaries}\label{sec:prelim}
\noindent\textbf{Probability.}
\begin{lemma}\label{eqn:move_permutation}
For any $ P_{\bfx_1,\cdots,\bfx_L}\in\Delta(\cX^L) $, any $ x_1,\cdots,x_L\in\cX $ and any $ \sigma\in S_L $, the following identity holds 
\begin{align}
P_{\bfx_1,\cdots,\bfx_L}(x_{\sigma(1)},\cdots,x_{\sigma(L)}) = P_{\bfx_{\sigma^{-1}(1)}, \cdots,\bfx_{\sigma^{-1}(L)}}(x_1,\cdots,x_L).
\notag
\end{align}
\end{lemma}

\begin{lemma}[Markov's inequality]\label{lem:markov}
If $X$ is a nonnegative random variable, then for any $a>0$, $ \prob{X\ge a}\le \expt{X}/a $. 
\end{lemma}

\begin{lemma}[Chebyshev's inequality]\label{lem:cheby}
If $X$ is an integrable random variable with finite expectation and finite nonzero variance, then for any $a>0$, $ \prob{\abs{X_\expt{X}}\ge a}\le\var{X}/a^2 $.
\end{lemma}

\begin{lemma}[Sanov's theorem]
\label{thm:sanov}
Let $\cQ\subset\Delta\paren{\cX}$ be a subset of distributions such that it is equal to the closure of its interior. Let $\vbfx\sim P_\bfx^{\otimes n}$ for some $P_\bfx\in\Delta(\cX)$. Note that $\expt{\tau_{\vbfx}} = P_\bfx$. Sanov's theorem determines the first-order exponent of the probability that the vector empirically looks like drawn from some distribution $Q\in\cQ$,
\[-\frac{1}{n}\log\prob{\tau_{\vbfx}\in\cQ}= \inf_{Q\in\cQ}D\paren{Q\|P_\bfx} \pm o_n(1).\]
\end{lemma}

\noindent\textbf{Channel coding.}
\begin{definition}[Oblivious AVMAC]
An oblivious AVMAC $ \cA = (\cX,\cY,\cS,\cZ,f_1,f_2,\Gamma_1,\Gamma_2,g,\Lambda,W_{\bfz|\bfx,\bfy,\bfs}) $ is a probability distribution $ W_{\bfz|\bfx,\bfy,\bfs} $ such that for every $ x\in\cX,y\in\cY,s\in\cS,z\in\cZ $,
\begin{align}
\prob{\bfz = z|\bfx = x,\bfy = y,\bfz = z} =& W_{\bfz|\bfx,\bfy,\bfs}(z|x,y,s).
\notag
\end{align}
If the users use the channel for $n\in\bZ_{>0} $ times, the channel acts on the transmitted sequences i.i.d., i.e., for any $ \vx\in\cX^n,\vy\in\cY^n,\vs\in\cS^n,\vz\in\cZ^n $,
\begin{align}
\prob{\vbfz = \vz|\vbfx = \vx,\vbfy = \vy,\vbfs = \vs} =& W_{\bfz|\bfx,\bfy,\bfs}^{\otimes n}(\vz|\vx,\vy,\vs) = \prod_{i=1}^nW_{\bfz|\bfx,\bfy,\bfs}(\vz(i)|\vx(i),\vy(i),\vs(i)).
\notag
\end{align}
Here the state sequence $\vs\in\cS^n $ is the output of James' jamming function $\jam$ which maps the codebook  pair of user one and two to a sequence $ \vs = \vs(\cC_1,\cC_2)\in\cS^n $ such that $ g(\vs)\le\Lambda $.
\end{definition}

\begin{definition}[Deterministic $L$-list-decodable code]
A deterministic $L$-list decodable code $ (\enc_1,\enc_2,\dec) $ for an oblivious AVMAC $ \cA = (\cX,\cY,\cS,\cZ,f_1,f_2,\Gamma_1,\Gamma_2,g,\Lambda,W_{\bfz|\bfx,\bfy,\bfs}) $ consists of 
\begin{itemize}
  \item an encoder for user one:
  \begin{align}
  \begin{array}{rlll}
  \enc_1\colon& \cM& \to& \cX^n \\
  &m&\mapsto& \vx_m
  \end{array},
  \notag
  \end{align}
  where $ \vx_m $ satisfies $ f_1(\vx_m)\le\Gamma_1 $ for all $m\in\cM$;
  \item an encoder for user two:
  \begin{align}
  \begin{array}{rlll}
  \enc_1\colon& \cW& \to& \cY^n \\
  &w&\mapsto& \vy_w
  \end{array},
  \notag
  \end{align}
  where $ \vy_w $ satisfies $ f_2(\vy_w)\le\Gamma_2 $ for all $w\in\cW$;
  \item a list decoder for the receiver:
  \begin{align}
  \begin{array}{rlll}
  \dec\colon& \cZ^n&\to& \binom{\cM\times\cW}{\le L} \\
  &\vz& \mapsto& \cL
  \end{array},
  \notag
  \end{align}
  where $ \cL\ni(m,w) $.
\end{itemize}
The dimension $n$ is called the blocklength of the code.

Let $ M\coloneqq|\cM| $ and $ W\coloneqq|\cW| $. 
The message sets $ \cM $ and $ \cW $ are identified with $ \sqrbrkt{M} $ and $ \sqrbrkt{W} $, respectively. 
The rate of a code $ (\cC_1,\cC_2) $ is defined as a pair $ (R_1,R_2) $ where $ R_1 = R(\cC_1) \coloneqq\frac{1}{n}\log (M/L) $ and $ R_2 = R(\cC_2) \coloneqq\frac{1}{n}\log (W/L) $. 

At times, we also abuse the notation and call the collection of codewords (images of the encoding maps) codebooks, i.e., $ \cC_1\coloneqq\curbrkt{\vx_{m}}_{m = 1}^M $, and $ \cC_2 \coloneqq\curbrkt{\vy_{w}}_{w = 1}^{W} $. 
\end{definition}

\begin{definition}[Average probability of error]
The average probability of error of a codebook pair $ (\cC_1,\cC_2) $ equipped with $ (\enc_1,\enc_2,\dec) $ for an oblivious AVMAC $ \cA = (\cX,\cY,\cS,\cZ,f_1,f_2,\Gamma_1,\Gamma_2,g,\Lambda,W_{\bfz|\bfx,\bfy,\bfs}) $ is defined as
\begin{align}
P_{\e,\avg}(\cC_1,\cC_2) \coloneqq& \max_{\substack{\vs = \vs(\cC_1,\cC_2)\\g(\vs)\le\Lambda}} \probover{\substack{\bfm\sim\cM\\\bfw\sim\cW}}{(\wh\bfm,\wh\bfw)\ne(\bfm,\bfw)} \notag \\
=& \max_{\substack{\vs = \vs(\cC_1,\cC_2)\\g(\vs)\le\Lambda}}\probover{\substack{\bfm\sim\cM\\\bfw\sim\cW}}{\dec(\vbfz)\ne(\bfm,\bfw)} \notag \\
=&  \max_{\substack{\vs = \vs(\cC_1,\cC_2)\\g(\vs)\le\Lambda}}\frac{1}{MW}\sum_{\substack{m\in\cM\\w\in\cW}}W_{\bfz|\bfx,\bfy,\bfs}^{\otimes n}(\vz|\vx_m,\vy_w,\vs)\indicator{(m,w)\notin\dec(\vz)}, \notag
\end{align}
where the probability is taken over uniform selection of $ \bfm $ and $ \bfw $.
\end{definition}

\begin{definition}[Achievable rate]
A rate pair $ (R_1,R_2) $ is said to be achievable  for an oblivious AVMAC  if for any constant $ \delta_1,\delta_2 >0 $ and $ \eps_1,\eps_2>0 $, there exists a sequence of codes $ \curbrkt{(\cC_{1,n},\cC_{2,n})}_{n} $ equipped with $ (\enc_{1,n},\enc_{2,n},\dec_n) $ for infinitely many $n  $ such that, there is an $ n_0 $, for every $ n>n_0 $,
\begin{itemize}
  \item   $ R_{1,n} \ge R_1 - \delta_1 $ and $ R_{2,n} \ge R_2 - \delta_2 $;
  \item the probabilities of user one's and user two's decoding errors vanish in $ n $, 
  \begin{align}
  P_{\e,A}(\cC_{1,n},\cC_{2,n}) \le& \eps_1, \notag \\
  P_{\e,B}(\cC_{1,n},\cC_{2,n}) \le& \eps_2. \notag 
  \end{align}
\end{itemize}
\end{definition}

\begin{definition}[$L$-list-decoding capacity]
The capacity $ (C_1,C_2) $ of an oblivious AVMAC is defined as the supremum of all achievable rates,
\begin{align}
C_1 \coloneqq& \limsup_{\eps\downarrow0}\limsup_{n\uparrow\infty}\max_{\substack{\cC_{1,n},\cC_{2,n}\\P_{\e,\avg}(\cC_1,\cC_2)\le\eps}}R(\cC_{1,n}), \notag \\
C_2 \coloneqq& \limsup_{\eps\downarrow0}\limsup_{n\uparrow\infty}\max_{\substack{\cC_{1,n},\cC_{2,n}\\P_{\e,\avg}(\cC_1,\cC_2)\le\eps}}R(\cC_{2,n}), \notag
\end{align}
where $ \cC_{1,n} $ and $ \cC_{2,n} $ satisfy power constraints.
\end{definition}

\noindent\textbf{Method of types.}
Without loss of generality, we write $\cX=\curbrkt{1,\cdots,{\card{\cX}}}$. For $\vx\in\cX^n$ and $x\in\cX$, let
\[N_x(\vx)\coloneqq\card{\curbrkt{i\in[n]\colon \vx(i)=x}},\]
which counts the number of occurrences of a symbol $x$ in a vector $\vx$. Similarly, define
\[N_{x,y}\paren{\vx,\vy}\coloneqq\card{\curbrkt{i\in[n]\colon \vx(i) = x,\;\vy(i) = y}}. \]

\begin{definition}[Types]
For a length-$n$ vector $\vx$ over a finite alphabet $\cX$, the type $\tau_{\vx}$ of $\vx$ is a length-$\card{\cX}$ (empirical) probability vector (or the histogram of $\vx$), i.e., $\tau_{\vx}\in[0,1]^{\card{\cX}}$ has entries
$\tau_{\vx}(x)\coloneqq{N_x(\vx)}/{n}$
for all $x\in\cX$.
\end{definition}

\begin{definition}[Joint types and conditional types]
The \emph{joint type} $\tau_{\vx,\vy}\in[0,1]^{\card{\cX}\times\card{\cY}}$ of two vectors $\vx\in\cX^n$ and $\vy\in\cY^n$ is defined as
$\tau_{\vx,\vy}(x,y)={N_{x,y}(\vx,\vy)}/{n}$
for $x\in\cX$ and $y\in\cY$.

The conditional type $\tau_{\vy|\vx}\in[0,1]^{\card{\cX}\times\card{\cY}}$ of a vector $\vy\in\cY^n$ given another vector $\vx\in\cX^n$ is defined as 
$\tau_{\vy|\vx}(y|x) = {N_{x,y}\paren{\vx,\vy}}/{N_x\paren{\vx}}$.
\end{definition}

\begin{remark}
We will also write $\tau_\bfx,\tau_{\bfx,\bfy},\tau_{\bfy|\vx}, \tau_{\bfy|\bfx}$ etc. for generic types that are taken from the corresponding sets of types even if they do not come from  instantiated vectors. For instance, $\tau_\bfx$ is a type  corresponding to any $\vx$ of that type. The particular choice of $\vx$ is not important and will not be specified. These notations are for  explicitly distinguishing  types from distributions.
\end{remark}

\begin{lemma}\label{lem:poly_many_types}
For $L$ ($L$ is a constant) finite sets $ \cX_1,\cdots,\cX_L $ of sizes independent of $n$, the number of types of  $L$-tuple of length-$n$ vectors $ (\vx_1,\cdots,\vx_L) $, where $ \vx_i\in\cX_i^n $ ($ 1\le i\le L $), is $ n^{\cO(1)} $.
\end{lemma}

\section{List decoding oblivious AVMACs with input and state constraints} \label{sec:obli_avmac_constrained}

In this section, we prove our main theorems.
\begin{theorem}[Achievability/inner bound]
If $ L>L_{w}^* $, then   any rate pair $ (R_1,R_2) $ in the interior of $ C_{\inner} $ is achievable. 
That is, for any $ \delta_1,\delta_2>0 $, there exists an $ L $-list decodable  code (sequence) $ (\cC_1,\cC_2) = ( \enc_1,\enc_2,\dec ) $ of rate $ (R_1 - \delta_1, R_2 - \delta_2) $ and vanishing (in $n$) average probability of error such that $ \card{\dec(\vz)}\le L $ for any $ \vz\in\cZ^n $. 
\end{theorem}

\begin{theorem}[Converse]
If $ L\le L_s^* $, then the $L$-list decoding capacity region is $ \curbrkt{(0,0)} $. 
That is, for any $ \eps_1,\eps_2>0 $ any  code $ (\cC_1,\cC_2) = (\enc_1,\enc_2,\dec) $ of  rate $ (\eps_1,\eps_2) $ such that $ \card{\dec(\vz)}\le L $ for any $ \vz\in\cZ^n $ must have average probability of error at least some positive constant. 
\end{theorem}

\subsection{Decoding rules}  
Given a codebook pair $\cC_1 = \curbrkt{\vx_{m}}_{m = 1}^{ L2^{nR_1}}$ and $\cC_2 = \curbrkt{\vy_w}_{w = 1}^{ L2^{nR_2}}$, and a time-sharing sequence $\vu$. 
Fix slack factors $\eta,\eta'>0$. 
For $\eta>0$, define the set of joint distributions that are consistent with the physical transmission across the channel
\[\cP_\eta \coloneqq \curbrkt{P_{\bfu,\bfx,\bfy,\bfs,\bfz}\in\Delta(\cU\times\cX\times\cY\times\cS\times\cZ)\colon \begin{array}{rl}
     D\paren{ P_{\bfu,\bfx,\bfy,\bfs,\bfz} \|  P_\bfu P_{\bfx|\bfu}P_{\bfy|\bfu}P_{\bfs|\bfu} W_{\bfz|\bfx,\bfy,\bfs} }\le&\eta,  \\
     \expt{g(\bfs)}\le&\Lambda
\end{array}  }.\]

Observing $\vz$, output all $(m,w)$ such that there exists $\vs$ with $g(\vs)\le\Lambda$ satisfying:
\begin{enumerate}
\item For $(\bfu,\bfx,\bfy,\bfs,\bfz)\sim \tau_{\vu,\vx_m,\vy_w,\vs,\vz}$, we have $P_{\bfu,\bfx,\bfy,\bfs,\bfz}\in\cP_\eta$;
\item For any bipartite graph $\cL = (\cS,\cT,\cF)\in\sL_L(m,w)$  and the corresponding list $\curbrkt{(\vx_{m'},\vy_{w'})\colon (m',w')\in\cF}$ such that for each $(m',w')\in\cF$, there exists $\vs_{m',w'}$ with $g(\vs_{m',w'})\le\Lambda$,  $P_{\bfu,\bfx_{m'},\bfy_{w'},\bfs_{m',w'},\bfz}\in\cP_\eta$, we have that
\begin{equation*}
    I\paren{ \left.\bfx,\bfy,\bfz;\bfx^{\cS\setminus m},\bfy^{\cT\setminus w}\right|\bfu,\bfs }\le\eta'.
\end{equation*}
\end{enumerate}

\subsection{Codebook construction} 
Codewords with the following desired properties can be obtained via random selection.
The proof is along the line of \cite{csiszar-narayan-it1988-obliviousavc} and we omit the details.

Let
\begin{align*}
    \cP_1 \coloneqq &\curbrkt{P_\bfx\in\Delta(\cX)\colon \expt{f_1(\bfx)} = \sum_x P_\bfx(x)f_1(x)\le\Gamma_1},\\
    \cP_2 \coloneqq &\curbrkt{P_{\bfy}\in\Delta(\cY)\colon\expt{f_2(\bfy)} = \sum_yP_\bfy(y)f_2(y)\le\Gamma_2}.
\end{align*}
\begin{lemma}
Fix any $\eps>0$, sufficiently large $n$, rate pair $R_1>\eps,R_2>\eps$,  types $P_\bfu,P_{\bfx|\bfu},P_{\bfy|\bfu}$ with $\sqrbrkt{P_\bfu P_{\bfx|\bfu}}_\bfx\in\cP_1,\sqrbrkt{P_\bfu P_{\bfy|\bfu}}_\bfy\in\cP_2$ and bipartite graph $\cB = (\cI,\cJ,\cE)\in\sB_L$, there exist a time-sharing sequence $\vu$ of type $P_\bfu$ and a codebook pair $\cC_1 = \curbrkt{ \vx_m }_{m=1}^{ L2^{nR_1}}, \cC_2 = \curbrkt{\vy_w}_{w=1}^{ L2^{nR_2}}$ of type $\tau_{\vx_m,\vy_w|\vu} = P_{\bfx|\bfu}P_{\bfy|\bfu}$ ($1\le m\le  L2^{nR_1},1\le w\le  L2^{nR_2}$) such that for every $\vx = \vx_{m_0}\in\cC_1$, $\vy = \vy_{w_0}\in\cC_2$, $\vs$ with $g(\vs)\le\Lambda$ and every joint type $P_{\bfu,\bfx^{\cI},\bfy^{\cJ},\bfs}$ with $P_{\bfx_i,\bfy_j|\bfu} = P_{\bfx|\bfu}P_{\bfy|\bfu}$ for every $(i,j)\in\cE$, the following properties hold. For every $(i,j)\in\cE$,
\begin{equation}
    \card{\curbrkt{  (m,w)\in\sqrbrkt{ L2^{nR_1}}\times\sqrbrkt{ L2^{nR_2}}\colon \tau_{\vu,\vx_m,\vy_w,\vs} = P_{\bfu,\bfx,\bfy,\bfs} }}\le 2^{n(R_1+R_2 - \eps/2)},
    \label{eqn:cw_selection_xy_s}
\end{equation}
if $I(\bfx,\bfy;\bfs|\bfu)\ge\eps$;
\begin{align}
    \card{\curbrkt{ (m',w')\colon \tau_{\vu,\vx,\vy,\vx_{m'},\vy_{w'},\vs} = P_{\bfu,\bfx,\bfy,\bfx_i,\bfy_{j},\bfs} }}\le&2^{n\paren{\sqrbrkt{R_1+R_2 - I(\bfx_i,\bfy_j;\bfx,\bfy,\bfs|\bfu)}^++\eps}};
    \label{eqn:cw_selection_no_mprimewprime}\\
    \card{\curbrkt{m'\in[M]\colon \tau_{\vu,\vx,\vy,\vx_{m'},\vs} = P_{\bfu,\bfx,\bfy,\bfx_i,\bfs}}}\le&2^{n\paren{\sqrbrkt{R_1-I(\bfx_i;\bfx,\bfy,\bfs|\bfu)}^++\eps}};\label{eqn:cw_selection_no_mprime}\\
    \card{\curbrkt{w'\in[W]\colon \tau_{\vu,\vx,\vy,\vy_{w'},\vs} = P_{\bfu,\bfx,\bfy,\bfy_j,\bfs} }}\le&2^{\paren{\sqrbrkt{R_1 - I(\bfy_j;\bfx,\bfy,\bfs|\bfu)}^++\eps}};\label{eqn:cw_selection_no_wprime}
\end{align}
and
\begin{equation}
    \card{\curbrkt{ (m,w)\colon \tau_{\vu,\vx_m,\vy_w,\vx_{m'},\vy_{w'},\vs} = P_{\bfu,\bfx,\bfy,\bfx_i,\bfy_j,\bfs},\text{ for some }m'\ne m,w'\ne w }}\le2^{n(R_1+R_2-\eps/2)},
    \label{eqn:cw_selection_takepositivepart}
\end{equation}
if $I(\bfx,\bfy;\bfx_i,\bfy_j,\bfs|\bfu) - \sqrbrkt{R_1+R_2 - I(\bfx_i,\bfy_j;\bfs|\bfu)}^+\ge\eps$;
\begin{align}
    \card{\curbrkt{ (m,w)\colon \tau_{\vu,\vx_m,\vy_w,\vx_{m'},\vs} = P_{\bfu,\bfx,\bfy,\bfx_i,\bfs},\text{ for some }m'\ne m }}\le&2^{n(R_1-\eps/2)},\label{eqn:cw_selection_takepositivepart_usr1}
\end{align}
if $I(\bfx,\bfy;\bfx_i,\bfs|\bfu) - \sqrbrkt{R_1 - I(\bfx_i;\bfs|\bfu)}^+\ge\eps$;
\begin{align}
    \card{\curbrkt{ (m,w)\colon \tau_{\vu,\vx_m,\vy_w,\vy_{w'},\vs} = P_{\bfu,\bfx,\bfy,\bfy_j,\bfs},\text{ for some }w'\ne w }}\le&2^{n(R_2-\eps/2)},
\end{align}
if $I(\bfx,\bfy;\bfy_j,\bfs|\bfu)-\sqrbrkt{R_2 - I(\bfy_j;\bfs|\bfu)}^+\ge\eps$.

Furthermore, if $R_1+R_2<\min_{(i,j)\in\cE}I(\bfx_i,\bfy_j;\bfs|\bfu)$, then 
\begin{align}
    \card{\curbrkt{ \cL = (\cS,\cT,\cF)\in\sL_L(m_0,w_0)\colon \begin{array}{l}
         \cL\text{ has the same underlying graph as }\cB,  \\
         \tau_{\vu,\vx_{m_0},\vy_{w_0},\vx^{\cS\setminus m_0},\vy^{\cT\setminus w_0},\vs} = P_{\bfu,\bfx,\bfy,\bfx^{I-1},\bfy^{J-1},\bfs}
    \end{array} }}\le&2^{n\eps}; \label{eqn:number_of_bipgh}
\end{align}
and
\begin{align}
    \card{\curbrkt{ (m,w)\in[M]\times[W]\colon \begin{array}{l}
         \tau_{\vu,\vx_m,\vy_w,\vx^{\cS\setminus m},\vy^{\cT\setminus w},\vs} = P_{\bfu,\bfx,\bfy,\bfx^{I-1},\bfy^{J-1},\bfs},  \\
          \text{ for some }\cL \in\sL_L(m,w)\text{ with the same graph structure as }\cB
    \end{array} }}\le&2^{n(R_1+R_2-\eps/2)},\label{eqn:cw_selection_xy_slistxy}
\end{align}
if $I(\bfx,\bfy;\bfx^{I-1},\bfy^{J-1},\bfs|\bfu)\ge\eps$.

If $R_1<\min_{i\in[L]} I(\bfx_i;\bfs|\bfu)$, then
\begin{align}
    \card{\curbrkt{ \cS\in\binom{[M]}{L}\colon \cS\ni m_0,\;\tau_{\vu,\vx_{m_0},\vy_{w_0},\vx^{\cS\setminus m_0},\vs} = P_{\bfu,\bfx,\bfy,\bfx^{L-1},\bfs} }}\le&2^{n\eps};\label{eqn:number_of_s}
\end{align}
and 
\begin{align}
    \card{\curbrkt{ (m,w)\colon \tau_{\vu,\vx_m,\vy_w,\vx^{\cS\setminus m},\vs} = P_{\bfu,\bfx,\bfy,\bfx^{L-1},\bfs},\text{ for some }\cS\text{ with }|\cS| =  L,\;\cS\ni m_0 }}\le&2^{n(R_1 - \eps/2)},\label{eqn:cw_selection_xy_slistx}
\end{align}
if $I(\bfx,\bfy;\bfx^{L-1},\bfs|\bfu)\ge\eps$.

If $R_2<\min_{j\in[L]}I(\bfy_j;\bfs|\bfu)$, then
\begin{align}
    \card{\curbrkt{ \cT\in\binom{[W]}{L}\colon \cT\ni w_0,\;\tau_{\vu,\vx_{m_0},\vy_{w_0},\vy^{\cT\setminus w_0},\vs} = P_{\bfu,\bfx,\bfy,\bfy^{L-1},\bfs} }}\le&2^{n\eps};
\end{align}
and
\begin{align}
    \card{\curbrkt{ (m,w)\colon\tau_{\vu,\vx_m,\vy_w,\vy^{\cT\setminus w,},\vs} = P_{\bfu,\bfx,\bfy,\bfy^{L-1},\bfs},\text{ for some }\cT\text{ with }|\cT| =  L,\;\cT\ni w_0 }}\le&2^{n(R_2-\eps/2)},
\end{align}
if $I(\bfx,\bfy;\bfy^{L-1},\bfs|\bfu)\ge\eps$.
\end{lemma}
\begin{remark}
When we say two graphs have the same structure, the equivalence is sensitive to vertex relabelling. 
Two bipartite graphs $ \cB = (\cI,\cJ,\cE)\in\sB_L $ and $ \cL = (\cS,\cT,\cF)\in\sL_L(m,w) $ have the same structure if $\cL$ is identical to $\cB$ after relabelling $\cS$ and $ \cT $ using $ \cI $ and $ \cJ $, respectively. Recall that we require that vertices in $ \cI,\cJ $ are consecutive increasing positive integers; vertices in $ \cS,\cT $ are messages of increasing indices. 
For example, in Fig. \ref{fig:graphs_with_same_struct}, $ \cB = (\cI,\cJ,\cE)\in\sB_4 $, where $ \cI = [3] $, $ \cJ=[2] $ and $ \cE = \curbrkt{(1,2),(2,1),(2,2),(3,2)} $, and $ \cL = (\cS,\cT,\cF)\in\sL_4(m_3,w_1) $, where $ \cS = \curbrkt{m_2,m_3,m_7} $, $ \cT = \curbrkt{w_1,w_3} $ and $ \cF = \curbrkt{(m_2,w_3),(m_3,w_1),(m_3,w_3),(m_7,w_3)} $, have the same structure.
However $ \cB $ does not have the same structure as $ \cL' = (\cS',\cT',\cF')\in\sL_4(m_3,w_1) $, where $ \cS' = \curbrkt{m_1,m_3,m_5} $, $ \cT' = \curbrkt{w_1,w_2} $ and $ \cF' = \curbrkt{(m_1,w_1),(m_3,w_1),(m_3,w_2),(m_5,w_1)} $, though they are isomorphic.
\begin{figure}[htbp]
  \centering
  \includegraphics[width=0.5\textwidth]{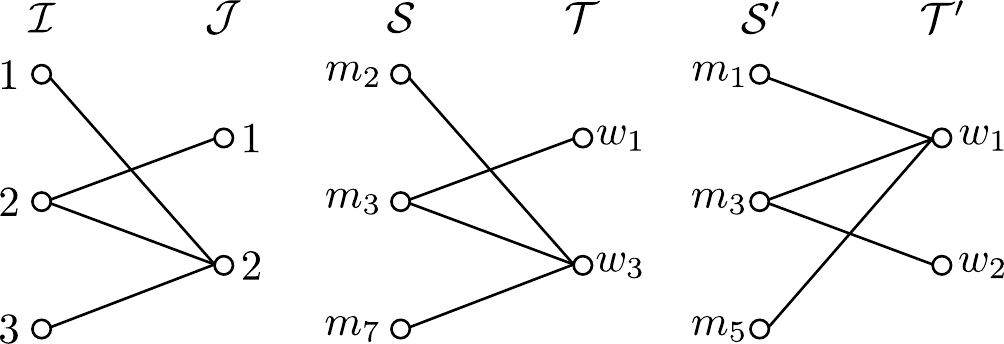}
  \caption{Graphs with the same underlying structure.}
  \label{fig:graphs_with_same_struct}
\end{figure}
\end{remark}

\subsection{Unambiguity of decoding}

\begin{lemma}\label{lem:unamb_dec}
Fix types $P_\bfu,P_{\bfx|\bfu},P_{\bfy|\bfu}$ with $\sqrbrkt{P_\bfu P_{\bfx|\bfu}}_\bfx\in\cP_1,\sqrbrkt{P_\bfu P_{\bfy|\bfu}}_\bfy\in\cP_2$. Fix any time-sharing sequence $\vu$ of type $P_\bfu$ and any codebook pair $\cC_1=\curbrkt{\vx_m}_{m=1}^{ L2^{nR_1}},\cC_1=\curbrkt{\vy_w}_{w=1}^{ L2^{nR_2}}$ such that $\tau_{\vx_m|\vu} = P_{\bfx|\bfu},\tau_{\vy_w|\vu} = P_{\bfy|\bfu}$. Assume that $P_\bfu(u)>0$ for all $u$ and $\overline\Lambda(\cB,P_\bfu P_{\bfx|\bfu}P_{\bfy|\bfu})>\Lambda$ for all bipartite graphs $\cB =(\cI,\cJ,\cE)\in\sB_L$. 

Suppose  $L> L_w(P_{\bfu,\bfx,\bfy})$.\footnote{In fact, it suffices to prove the lemma for $L=L_w(P_{\bfu,\bfx,\bfy})+1$.} Then the decoder defined above always outputs a list  of at most $L$ message pairs. That is, there is no bipartite graph $\cB = (\cI,\cJ,\cE)\in\sB_{L+1}$ and no joint distribution $P_{\bfu,\bfx^\cI,\bfy^\cJ,\bfs^\cE,\bfz}$ simultaneously satisfying
\begin{enumerate}
    \item $P_{\bfx_i|\bfu} = P_{\bfx|\bfu},P_{\bfy_j|\bfu} = P_{\bfy|\bfu}$ for all $i\in\cI,j\in\cJ$;
    \item $\expt{g(\bfs_{i,j})}\le\Lambda$ for all $(i,j)\in\cE$;
    \item $P_{\bfu,\bfx_i,\bfy_j,\bfs_{i,j},\bfz}\in\cP_\eta$ for all $(i,j)\in\cE$;
    \item $I\paren{\left. \bfx_i,\bfy_j,\bfz;\bfx^{\cI\setminus i},\bfy^{\cJ\setminus j} \right| \bfu,\bfs_{i,j} }\le\eta'$ for all $(i,j)\in\cE$.
\end{enumerate}
\end{lemma}

\begin{proof}
The proof is by contradiction. Suppose that there is a bipartite graph $\cB = (\cI,\cJ,\cE)\in\sB_{L+1}$ and a joint distribution $P_{\bfu,\bfx^\cI,\bfy^\cJ,\bfs^\cE,\bfz}$ satisfying the above conditions. 

Now consider the divergences, for $(i,j)\in\cE$,
\begin{equation*}
    D\paren{ P_{\bfu,\bfx_i,\bfy_j,\bfx^{\cI\setminus i},\bfy^{\cJ\setminus j},\bfs_{i,j},\bfz} \left\| P_\bfu P_{\bfx_i|\bfu}P_{\bfy_j|\bfu}P_{\bfx^{\cI\setminus i},\bfy^{\cJ\setminus j},\bfs_{i,j}|\bfu}W_{\bfz|\bfx_i,\bfy_j,\bfs_{i,j}} \right.}.
\end{equation*}

One can verify that the above divergence is the sum of 
\begin{align*}
    \eta\ge &D\paren{ P_{\bfu,\bfx_i,\bfy_j,\bfs_{i,j},\bfz} \left\| P_{\bfu}P_{\bfx_i|\bfu}P_{\bfy_j|\bfu}P_{\bfs_{i,j}W_{\bfz|\bfx_i,\bfy_j,\bfs_{i,j}}} \right.}\\
    =&\sum_{u,x^\cI,y^\cJ,s,z}P(u,x_i,y_j,x^{\cI\setminus i},y^{\cJ\setminus j},s,z)\log\frac{P(u,x_i,y_j,s,z)}{P(u)P(x_i|u)P(y_j|u)P(s)W(z|x_i,y_j,s)}
\end{align*}
and
\begin{align*}
    \eta'\ge&I\paren{\left. \bfx_i,\bfy_j,\bfz;\bfx^{\cI\setminus i},\bfy^{\cJ\setminus j} \right|\bfu,\bfs_{i,j}  }\\
    =&\sum_{u,x^\cI,y^\cJ,s,z} P(u,x_i,y_j,x^{\cI\setminus i},y^{\cJ\setminus j},s,z) \log\frac{P(x^{\cI\setminus i},y^{\cJ\setminus j}|u,s,x_i,y_j,z)}{P(x^{\cI\setminus i},y^{\cJ\setminus j}|u,s)}.
\end{align*}
Hence each of  the above divergences is at most $\eta+\eta'$. Since marginalization does not increase divergence, we have
\begin{equation*}
    D\paren{\left.  P_{\bfu,\bfx_i,\bfy_j,\bfx^{\cI\setminus i},\bfy^{\cJ\setminus j},\bfz}  \right\| P_{\bfu}P_{\bfx_i|\bfu}P_{\bfy_j|\bfu}V_{\bfx^{\cI\setminus i},\bfy^{\cJ\setminus j},\bfz|\bfu,\bfx_i,\bfy_j} }\le\eta+\eta',
\end{equation*}
where
\begin{equation*}
    V_{\bfx^{\cI\setminus i},\bfy^{\cJ\setminus j},\bfz|\bfu,\bfx_i,\bfy_j}(x^{\cI\setminus i},y^{\cJ\setminus j},z|u,x_i,y_j) \coloneqq\sum_s P_{\bfx^{\cI\setminus i},\bfy^{\cJ\setminus j},\bfs_{i,j}|\bfu}(x^{\cI\setminus i},y^{\cJ\setminus j},s|u)W_{\bfz|\bfx_i,\bfy_j,\bfs_{i,j}}(z|x_i,y_j,s).
\end{equation*}
By Pinsker's inequality, the divergence is  lower bounded by the total variation distance (multiplied by some universal constant). We hence have
\begin{align*}
    c\sqrt{\eta+\eta'}\ge&\sum_{u,x^\cI,y^\cJ,z}\abs{ P(u,x_i,y_j,x^{\cI\setminus i},y^{\cJ\setminus j},z) - P(u)P(x_i|u)P(y_j|u)V(x^{\cI\setminus i},y^{\cJ\setminus j},z|u,x_i,y_j) }, 
\end{align*}
where $c=\sqrt{2\ln2}$.

Similarly, for $(i',j')\in\cE$, we have the same inequality
\begin{align*}
    c\sqrt{\eta+\eta'}\ge&\sum_{u,x^\cI,y^\cJ,z}\abs{ P(u,x_{i'},y_{j'},x^{\cI\setminus i'},y^{\cJ\setminus j'},z) - P(u)P(x_{i'}|u)P(y_{j'}|u)V(x^{\cI\setminus i'},y^{\cJ\setminus j'},z|u,x_{i'},y_{j'}) },
\end{align*}
where
\begin{align*}
    V_{\bfx^{\cI\setminus i'},\bfy^{\cJ\setminus j'},\bfz|\bfu,\bfx_{i'},\bfy_{j'}}'(x^{\cI\setminus i'},y^{\cJ\setminus j'},z|u,x_{i'},y_{j'}) \coloneqq\sum_s P_{\bfx^{\cI\setminus i'},\bfy^{\cJ\setminus j'},\bfs_{i',j'}|\bfu}(x^{\cI\setminus i'},y^{\cJ\setminus j'},s|u)W_{\bfz|\bfx_{i'},\bfy_{j'},\bfs_{i',j'}}(z|x_{i'},y_{j'},s).
\end{align*}
By triangle inequality,
\begin{align*}
    &2c\sqrt{\eta+\eta'}\\
    \ge&\sum_{u,x^\cI,y^\cJ,z}\abs{ P(u)P(x_i|u)P(y_j|u)V(x^{\cI\setminus i},y^{\cJ\setminus j},z|u,x_i,y_j) - P(u)P(x_{i'}|u)P(y_{j'}|u)V'(x^{\cI\setminus i'},y^{\cJ\setminus j'},z|u,x_{i'},y_{j'}) }\\
    =&\sum_{u,x^\cI,y^\cJ,z}\left|\sum_s P(u)P(x_i|u)P(y_j|u)P(x^{\cI\setminus i},y^{\cJ\setminus j},s|u)W(z|x_i,y_j,s)- P(u)P(x_{i'}|u)P(y_{j'}|u)P(x^{\cI\setminus i'},y^{\cJ\setminus j'},s|u)W(z|x_{i'},y_{j'},s) \right|.
\end{align*}
 Let $p_u^*\coloneqq\min_uP_\bfu(u)$. Since $P_\bfu$ is assumed to have no zero atom, $p_u^*>0$. By Markov's inequality,
\begin{align}
    &\frac{2c\sqrt{\eta+\eta'}}{p_u^*}\notag\\
    \ge&\sum_{x^\cI,y^\cJ,z}\left|\sum_{s} P(x_i|u)P(y_j|u)P(x^{\cI\setminus i},y^{\cJ\setminus j},s|u)W(z|x_i,y_j,s)- P(x_{i'}|u)P(y_{j'}|u)P(x^{\cI\setminus i'},y^{\cJ\setminus j'},s|u)W(z|x_{i'},y_{j'},s) \right|\label{eqn:summation_of_interest}.
\end{align}
Note that for any $\sigma\in S_{I-1},\pi\in S_{J-1}$, the summation \eqref{eqn:summation_of_interest} equals
\begin{align*}
    &\sum_{x^\cI,y^\cJ,z}\left|\sum_{s} P(x_i|u)P(y_j|u)P(x^{\sigma(\cI\setminus i)},y^{\pi(\cJ\setminus j)},s|u)W(z|x_i,y_j,s)- P(x_{i'}|u)P(y_{j'}|u)P(x^{\sigma(\cI\setminus i')},y^{\pi(\cJ\setminus j')},s|u)W(z|x_{i'},y_{j'},s) \right|.
\end{align*}
Hence the RHS of Eqn. \eqref{eqn:summation_of_interest} equals
\begin{align}
    &\frac{1}{(I-1)!(J-1)!}\sum_{\sigma\in S_{I-1},\pi\in S_{J-1}}\sum_{x^\cI,y^\cJ,z}\left|\sum_{s} P(x_{i}|u)P(y_{j}|u)P(x^{\sigma(\cI\setminus i)},y^{\pi(\cJ\setminus j)},s|u)W(z|x_{i},y_{j},s)\right.\notag\\
    &\left.-\sum_s P(x_{i'}|u)P(y_{j'}|u)P(x^{\sigma(\cI\setminus i')},y^{\pi(\cJ\setminus j')},s|u)W(z|x_{i'},y_{j'},s) \right|\notag\\
    \ge&\sum_{x^\cI,y^\cJ,z}\left|\sum_{s} P(x_i|u)P(y_j|u)Q(x^{\cI\setminus i},y^{\cJ\setminus j},s|u)W(z|x_i,y_j,s)- P(x_{i'}|u)P(y_{j'}|u)Q(x^{\cI\setminus i'},y^{\cJ\setminus j'},s|u)W(z|x_{i'},y_{j'},s) \right|,\label{eqn:wts_positive}
\end{align}
where 
\begin{align*}
    Q_{\bfx^{I-1},\bfy^{J-1},\bfs|\bfu}\coloneqq&\frac{1}{(I-1)!(J-1)!}\sum_{\sigma\in S_{I-1},\pi\in S_{J-1}}P_{\bfx^{\sigma^{-1}(\cI\setminus i)},\bfy^{\pi^{-1}(\cJ\setminus j)},\bfs_{i,j}|\bfu} .
\end{align*}
One can check that $ Q_{\bfx^{I-1},\bfy^{J-1},\bfs|\bfu} $ is symmetric in $ x^{I-1},y^{J-1} $ for every $ u,s $.
Indeed, for any $u,s,x^{I-1},y^{J-1} $ and $ \sigma'\in S_{I-1},\pi'\in S_{J-1} $,
\begin{align}
Q(x^{\sigma'(I-1)},y^{\pi'(J-1)},s|u) =& \frac{1}{(I-1)!(J-1)!}\sum_{\sigma\in S_{I-1},\pi\in S_{J-1}} P(x^{\sigma(\sigma'(I-1))},y^{\pi(\pi'(J-1))},s|u) \notag \\
=& \frac{1}{(I-1)!(J-1)!}\sum_{\sigma\in S_{I-1},\pi\in S_{J-1}}P(x^{\sigma(I-1)},y^{\pi(J-1)},s|u) \notag \\
=& Q(x^{I-1},y^{J-1},s|u). \notag
\end{align}
Let $f(Q_{\bfx^{I-1},\bfy^{J-1},\bfs|\bfu},P_{\bfx|\bfu},P_{\bfy|\bfu})$ denote the RHS of Eqn. \eqref{eqn:wts_positive} maximized over edges $(i',j')\ne(i,j)$.
Suppose that, via  distributions  $Q_{\bfx^{I-1},\bfy^{J-1},\bfs|\bfu}^*,P_{\bfx|\bfu}^*,P_{\bfy|\bfu}^*$, $f$ attains its maxima $f(Q^*,P_{\bfx|\bfu}^*,P_{\bfy|\bfu}^*)\eqcolon\zeta$. We will argue that $\zeta>0$. Assume otherwise $\zeta=0$. Then for all $(i',j')\ne(i,j),x^\cI,y^\cJ,z$,
\begin{align}
    \sum_{s} P^*(x_i|u)P^*(y_j|u)Q^*(x^{\cI\setminus i},y^{\cJ\setminus j},s|u)W(z|x_i,y_j,s) =&\sum_{s} P^*(x_{i'}|u)P^*(y_{j'}|u)Q^*(x^{\cI\setminus i'},y^{\cJ\setminus j'},s|u)W(z|x_{i'},y_{j'},s).\label{eqn:towards_contradiction_symm}
\end{align}
Marginalizing $z$ out, we have
\begin{align}
    P^*(x_i|u)P^*(y_j|u)Q^*(x^{\cI\setminus i},y^{\cJ\setminus j}|u)=&P^*(x_{i'}|u)P^*(y_{j'}|u)Q^*(x^{\cI\setminus i'},y^{\cJ\setminus j'}|u).
\end{align}
In fact, $Q$ satisfying the above identity must be a product distribution.
\begin{align}
    Q^*(x^{\cI\setminus i},y^{\cJ\setminus j}|u) = &(P^*)^{\otimes(I-1)}(x^{\cI\setminus i}|u)(P^*)^{\otimes(J-1)}(y^{\cJ\setminus j}|u),\label{eqn:prod_distr}
\end{align}
which is obviously symmetric
The proof of the above identity is deferred to Lemma \ref{lem:prod_distr_lem}.
Substituting this back to Eqn. \eqref{eqn:towards_contradiction_symm}
\begin{align*}
    &\sum_{s} P^*(x_i|u)P^*(y_j|u)Q^*(x^{\cI\setminus i},y^{\cJ\setminus j}|u)Q^*(s|u,x^{\cI\setminus i},y^{\cJ\setminus j})W(z|x_i,y_j,s)\\
    =&\sum_{s} P^*(x_{i'}|u)P^*(y_{j'}|u)Q^*(x^{\cI\setminus i'},y^{\cJ\setminus j'}|u)Q^*(s|u,x^{\cI\setminus i'},y^{\cJ\setminus j'})W(z|x_{i'},y_{j'},s),
\end{align*}
we have
\begin{align*}
    &\sum_s(P^*)^{\otimes I}(x^\cI|u)(P^*)^{\otimes J}(y^{\cJ}|u)Q^*(s|u,x^{\cI\setminus i},y^{\cJ\setminus j})W(z|x_i,y_j,s)\\
    =&\sum_s(P^*)^{\otimes I}(x^\cI|u)(P^*)^{\otimes J}(y^{\cJ}|u)Q^*(s|u,x^{\cI\setminus i'},y^{\cJ\setminus j'})W(z|x_{i'},y_{j'},s).
\end{align*}
Cancelling out $(P^*)^{\otimes I}(x^\cI|u)(P^*)^{\otimes J}(y^{\cJ}|u)$ which is independent of $s$, we get 
\begin{align*}
    \sum_sQ^*(s|u,x^{\cI\setminus i},y^{\cJ\setminus j})W(z|x_i,y_j,s)=&\sum_sQ^*(s|u,x^{\cI\setminus i'},y^{\cJ\setminus j'})W(z|x_{i'},y_{j'},s).
\end{align*}
Since 
\begin{align}
Q^*(x^{I-1},y^{J-1},s|u) =& Q^*(x^{I-1},y^{J-1}|u) Q^*(s|x^{I-1},y^{J-1},u),
\end{align}
and both $ Q^*(x^{I-1},y^{J-1},s|u) $ and $ Q^*(x^{I-1},y^{J-1}|u) $ are symmetric in $ x^{I-1},y^{J-1} $, $ Q^*(s|x^{I-1},y^{J-1},u) $ is also symmetric. 
Therefore, $\wt Q_u(s|x^{I-1},y^{J-1}) \coloneqq Q^*(s|u,x^{I-1},y^{J-1})$ is a symmetrizing distribution for every $u$. Note also that 
\begin{align}
    & \sum_{u,x^{I-1},y^{J-1}}P(u)(P^*)^{\otimes (I-1)}(x^{I-1}|u)(P^*)^{\otimes (J-1)}(y^{J-1}|u)\wt Q_u(s|x^{I-1},y^{J-1})g(s)\notag\\
    =&\sum_{u,x^{I-1},y^{J-1},s}P(u)P^*(x^{I-1}|u)P^*(y^{J-1}|u)Q^*(s|u,x^{I-1},y^{J-1})g(s)\notag\\
    =&\sum_{u,x^{I-1},y^{J-1},s}P(u)Q^*(x^{I-1},y^{J-1},s|u)g(s)\label{eqn:by_bayes}\\
    =&\frac{1}{(I-1)!(J-1)!}\sum_{\sigma\in S_{I-1},\pi\in S_{J-1}}\sum_{u,x^{I-1},y^{J-1},s}P(u)P_{\bfx^{\sigma^{-1}(\cI\setminus i)},\bfy^{\pi^{-1}(\cJ\setminus j)},\bfs_{i,j}|\bfu}(x^{I-1},y^{J-1},s|u)g(s)\notag\\
    =&\sum_{u,x^{I-1},y^{J-1},s}P(u)P_{\bfx^{\cI\setminus i},\bfy^{\cJ\setminus j},\bfs_{i,j}|\bfu}(x^{I-1},y^{J-1},s|u)g(s)\label{eqn:same_sum_for_diff_perm}\\
    =&\sum_sP_{\bfs_{i,j}}(s)g(s)\notag\\
    =&\expt{g(\bfs_{i,j})}\notag\\
    \le&\Lambda,\label{eqn:by_assumption}
\end{align}
where Eqn. \eqref{eqn:by_bayes} is by Bayes' theorem, 
\begin{align*}
    Q^*(s|u,x^{I-1},y^{J-1}) = &\frac{P(u)Q^*(x^{I-1},y^{J-1},s|u)}{P(u,x^{I-1},y^{J-1})}\\
    =&\frac{P(u)Q^*(x^{I-1},y^{J-1},s|u)}{P(u)P^*(x^{I-1}|u)P^*(y^{J-1}|u)}\\
    =&\frac{Q^*(x^{I-1},y^{J-1},s|u)}{P^*(x^{I-1}|u)P^*(y^{J-1}|u)}.
\end{align*}
Eqn. \eqref{eqn:same_sum_for_diff_perm} follows since the inner summation is invariant under every permutation pair $(\sigma,\pi)$. Eqn. \eqref{eqn:by_assumption} is by the assumption of this lemma.
We thus have found a family of symmetrizing distributions subject to power constraints which have an underlying graph $\cB\in\sB_{L+1}$ with $|\cE(\cB)|=L+1$ edges, which means $L_w(P_{\bfu,\bfx,\bfy})\ge L+1$. This contradicts the assumption $L_w(P_{\bfu,\bfx,\bfy})<L$ and finishes the proof.
\end{proof}

It remains to check Eqn. \eqref{eqn:prod_distr}.
\begin{lemma}\label{lem:prod_distr_lem}
For $I\ge2$, $J\ge2$, let $Q\in\Delta(\cX^{I-1}\times\cY^{J-1}|\cU)$ and $P_1\in\Delta(\cX|\cU),P_2\in\Delta(\cY|\cU)$ be such that
\begin{align}
    Q(x^{\cI\setminus i},y^{\cJ\setminus j}|u)P_1(x_i|u)P_2(y_j|u) = &Q(x^{\cI\setminus i'},y^{\cJ\setminus j'}|u)P_1(x_{i'}|u)P_2(y_{j'}|u)\label{eqn:prod_distr_assump}
\end{align}
for all $(i,j)\ne(i',j'),x^{\cI},y^{\cJ},u$. Then
\begin{align}
    Q(x^{\cI\setminus i},y^{\cJ\setminus j}|u)P_1(x_i|u)P_2(y_j|u) = P_1^{\otimes I}(x^\cI|u)P_2^{\otimes J}(y^\cJ|u),\label{eqn:prod_distr_to_show}
\end{align}
for all $i,j,x^\cI,y^\cJ$.
\end{lemma}
\begin{proof}
The proof is by induction on $ I $ and $ J $.
When $ I =  J=2$, i.e., $\cI = \curbrkt{i,i'},\cJ = \curbrkt{j,j'}$, Eqn. \eqref{eqn:prod_distr_assump} reduces to 
\begin{align*}
    Q(x_{i'},y_{j'}|u)P_1(x_i|u)P_2(y_j|u) = &Q(x_{i},y_{j}|u)P_1(x_{i'}|u)P_2(y_{j'}|u).
\end{align*}
Summing over $x_i,y_j$ on both sides, we get $Q(x_{i'},y_{j'}|u) = P_1(x_{i'}|u)P_2(y_{j'}|u)$. This proves Eqn. \eqref{eqn:prod_distr_to_show} for $ I =  J = 2$.

Assume that Eqn. \eqref{eqn:prod_distr_to_show} holds for $ I - 1$ and $J- 1$. 

For $(i,j)\ne(i',j')$ and $i,i'\ne I,j,j'\ne J$, summing over $x_{ I}$ and $y_{ J}$ on both sides of Eqn. \eqref{eqn:prod_distr_assump} yields
\begin{align*}
    Q(x^{\cI\setminus(i\cup I)},y^{\cJ\setminus (j\cup J)})P_1(x_i|u)P_2(y_j|u) = &Q(x^{\cI\setminus(i'\cup J)},y^{\cJ\setminus (j'\cup J)})P_1(x_{i'}|u)P_2(y_{j'}|u),
\end{align*}
which can also be written as
\begin{align*}
    Q(x^{[I-1]\setminus i},y^{[J-1]\setminus j})P_1(x_i|u)P_2(y_j|u) = &Q(x^{[I-1]\setminus i'},y^{[J-1]\setminus  j'})P_1(x_{i'}|u)P_2(y_{j'}|u).
\end{align*}
By induction hypothesis, 
\begin{align}
    Q(x^{[I-1]\setminus i},y^{[J-1]\setminus j})P_1(x_i|u)P_2(y_j|u)=&P_1^{\otimes (I-1)}(x^{I-1}|u)P_2^{\otimes(J-1)}(y^{J-1}|u).\label{eqn:induction_hyp_implies}
\end{align}

For $i'=I$ and $j'=J$, doing the same thing gives that, for $i\ne i',j\ne j'$,
\begin{align}
    Q(x^{[I-1]\setminus i},y^{[J-1]\setminus j}|u)P_1(x_i|u)P_2(y_j|u)
    =&Q(x^{\cI\setminus(i\cup I)},y^{\cJ\setminus(j\cup J)}|u)P_1(x_i|u)P_2(y_j|u) \notag\\
    =&Q(x^{[I-1]},y^{[J-1]}|u).\label{eqn:marginalize_last_coord}
\end{align}
Combining Eqn. \eqref{eqn:induction_hyp_implies} and Eqn. \eqref{eqn:marginalize_last_coord}, we have
\begin{align*}
    Q(x^{[I-1]},y^{[J-1]}|u)=&P_1^{\otimes (I-1)}(x^{I-1}|u)P_2^{\otimes(J-1)}(y^{J-1}|u),
\end{align*}
which finishes the proof.
\end{proof}

\subsection{Achievability}  
Define, for some bipartite graph $\cB = (\cI,\cJ,\cE)\in\sB_{L-1}$ and distribution $P_{\bfx,\bfy} = P_\bfx P_\bfy$,
\begin{equation*}
    \wt\Lambda(\cB,P_{\bfx,\bfy}) \coloneqq \min_{Q_{\bfs|\bfx^{I-1},\bfy^{J-1}}\in\cQ_\symm(\cB)}\sum_{x^{I-1},y^{J-1},s}P_{\bfx}^{\otimes (I-1)}(x^{I-1})P_{\bfy}^{\otimes (J-1)}(y^{J-1})Q_{\bfs|\bfx^{I-1},\bfy^{J-1}}(s|x^{I-1},y^{J-1})g(s).
\end{equation*}
For $P_{\bfy,\bfx,\bfy} = P_\bfu P_{\bfx|\bfu}P_{\bfy|\bfu}$, define
\begin{align*}
    \overline{\Lambda}(\cB,P_{\bfu,\bfx,\bfy})\coloneqq&\sum_uP_\bfu(u)\wt\Lambda(\cB,P_{\bfx,\bfy|u})\\
    =&\min_{\curbrkt{Q_{\bfs|\bfx^{I-1},\bfy^{J-1}}^{(u)}}\subset\cQ_\symm(\cB)}\sum_{u,x^{I-1},y^{J-1},s}P_\bfu(u)P_{\bfx|u}^{\otimes (I-1)}(x^{I-1})P_{\bfy|u}^{\otimes (J-1)}(y^{J-1})Q_{\bfs|\bfx^{I-1},\bfy^{J-1}}^{(u)}(s|x^{I-1},y^{J-1})g(s).
\end{align*}

Fix  $P_{\bfu},P_{\bfx|\bfu},P_{\bfy|\bfu}$ with $\sqrbrkt{P_\bfu P_{\bfx|\bfu}}_\bfx\in\cP_1,\sqrbrkt{P_\bfu P_{\bfy|\bfu}}_\bfy\in\cP_2$. 
Assume $P_\bfu(u)>0$ for every $u$.
We write $L_w$ and $L_s$ instead of $L_w(P_{\bfu}P_{\bfx|\bfy}P_{\bfy|\bfu})$ and $L_s(P_{\bfu}P_{\bfx|\bfy}P_{\bfy|\bfu})$ for brevity.
Assume $L>L_w$, e.g., $L = L_w+1$. 
We know, by non-symmetrizability, that $\wt\Lambda(\cB,P_{\bfu,\bfx,\bfy})>\Lambda$ for any $\cB\in\sB_L$. Fix $\eta,\eta'$ such that $\psi(\vz)\le L$ for all received $\vz$ under our decoder $\psi$ and 
\begin{align*}
    I(\bfx',\bfy';\bfz|\bfu)\ge&I(\bfx,\bfy;\bfz|\bfu)-\delta/3,\\
    I(\bfx';\bfz|\bfu)\ge&I(\bfx;\bfz|\bfu)-\delta/3,\\
    I(\bfy';\bfz|\bfu)\ge&I(\bfy;\bfz|\bfu)-\delta/3,
\end{align*}
for any $P_{\bfu,\bfx',\bfy',\bfs',\bfz}\in\cP_\eta$. 

Choose $R_1,R_2$ such that
\begin{align*}
    R_1 =& I(\bfx;\bfz|\bfu) - 5\delta/6,\\
    R_2 =& I(\bfy;\bfz|\bfu) - 5\delta/6,\\
    R_1+R_2=&I(\bfx,\bfy;\bfz|\bfu) - 5\delta/6.
\end{align*}
Let $M\coloneqq L2^{nR_1},W\coloneqq L2^{nR_2}$.
Let $\vu$ be a time-sharing sequence of type $P_\bfu$. Let $\eps = \min\curbrkt{\delta/5,2\eta/5}$. Pick a codebook pair $\cC_1 = \curbrkt{\vx_{m}}_{m=1}^{ L2^{nR_1}},\cC_2 = \curbrkt{\vy_w}_{w = 1}^{ L2^{nR_2}}$ which are $P_{\bfx|\vu}$- and $P_{\bfy|\vu}$-constant composition, respectively, and satisfy properties mentioned above.

Let $(m,w)$ be the transmitted message pair. Fix $\vs$ with $g(\vs)\le\Lambda$. Conditioning on $\cC_1,\cC_2,\vu,\vs,m,w$ is omitted for brevity.  We use boldface lower-case letters to denote random variables distributed according to the types of the corresponding vectors. We write $\bfx,\bfy$ instead of $\bfx_m,\bfy_w$, for short.

We now bound the average probability of error $ P_{e,\avg}(\vs)$ under the action of $\vs$.
A decoding error occurs either if $P_{\bfu,\bfx,\bfy,\bfs,\bfz}\notin\cP_\eta$ or if there is a bipartite graph $\cB\in\sB_L$ and a joint distribution $P_{\bfu,\bfx,\bfy,\bfx^{I-1},\bfy^{J-1},\bfs,\bfz}$ such that
\begin{enumerate}
    \item $P_{\bfu,\bfx,\bfy,\bfs,\bfz}\in\cP_\eta$;
    \item for each $(i,j)\in\cE$, there is an $\bfs_{i,j}$ with $\expt{\bfs_{i,j}}\le\Lambda$ such that $P_{\bfu,\bfx_i,\bfy_j,\bfs_{i,j},\bfz}\in\cP_\eta$;
    \item $I(\bfx,\bfy,\bfz;\bfx^{I-1},\bfy^{J-1}|\bfu,\bfs)>\eta'$.
\end{enumerate}

Define error events
\begin{align*}
    \cG_\atyp \coloneqq & \curbrkt{ D\paren{P_{\bfu,\bfx,\bfy,\bfs,\bfz}\|P_\bfu P_{\bfx|\bfu}P_{\bfy|\bfu}P_\bfs W_{\bfz|\bfx,\bfy,\bfs}}>\eta };\\
    \cG  \coloneqq& \curbrkt{I(\bfx,\bfy;\bfs|\bfu)\ge\eps};
\end{align*}
For $\cB\in\sB_L$ with $I>1,J>1$, let
\begin{align}
    \cD_{\eta,\eta'}(\cB)\coloneqq&\curbrkt{ \tau_{\bfu,\bfx,\bfy,\bfx^{I-1},\bfy^{J-1},\bfs,\bfz}\colon \begin{array}{rl}
        \tau_{\bfu,\bfx,\bfy,\bfs,\bfz}\in&\cP_\eta, \\
        \forall(i,j)\in\cE,\;\exists\bfs_{i,j},\;\expt{\bfs_{i,j}}\le\Lambda,\;\tau_{\bfu,\bfx_i,\bfy_j,\bfs_{i,j},\bfz}\in&\cP_\eta,\\
        I(\bfx,\bfy,\bfz;\bfx^{I-1},\bfy^{J-1}|\bfu,\bfs)>&\eta'
    \end{array} }.\notag
\end{align}
Also define
\begin{align}
    \cD_{\eta,\eta'}'\coloneqq&\curbrkt{ \tau_{\bfu,\bfx,\bfy,\bfx^{L-1},\bfs,\bfz}\colon \begin{array}{rl}
        \tau_{\bfu,\bfx,\bfy,\bfs,\bfz}\in&\cP_\eta, \\
        \forall i\in[L],\;\exists\bfs_{i},\;\expt{\bfs_{i}}\le\Lambda,\;\tau_{\bfu,\bfx_i,\bfy,\bfs_{i},\bfz}\in&\cP_\eta,\\
        I(\bfx,\bfy,\bfz;\bfx^{I-1}|\bfu,\bfs)>&\eta'
    \end{array} },\notag\\
    \cD_{\eta,\eta'}''\coloneqq&\curbrkt{ \tau_{\bfu,\bfx,\bfy,\bfy^{L-1},\bfs,\bfz}\colon \begin{array}{rl}
        \tau_{\bfu,\bfx,\bfy,\bfs,\bfz}\in&\cP_\eta, \\
        \forall j\in[L],\;\exists\bfs_{j},\;\expt{\bfs_{j}}\le\Lambda,\;\tau_{\bfu,\bfx,\bfy_j,\bfs_{j},\bfz}\in&\cP_\eta,\\
        I(\bfx,\bfy,\bfz;\bfy^{L-1}|\bfu,\bfs)>&\eta'
    \end{array} }.\notag
\end{align}
We write $\cD(\cB),\cD',\cD''$ for short. 

Let 
\begin{align}
    \cD=\cD_{\eta,\eta'}\coloneqq&\bigcup_{\substack{\cB\in\sB_L\\I>1,J>1}}\cD_{\eta,\eta'}(\cB),\notag\\
    e(m,w,\vs)\coloneqq&\sum_{\vz\colon (m,w)\notin\psi(\vz)}W(\vz|\vx_m,\vy_w,\vs).\notag
\end{align}
The average probability of error under the action of  a given $\vs$ is 
\begin{align}
    P_{e,\avg}(\vs) = &\frac{1}{MW}\sum_{(m,w)}e(m,w,\vs)\notag\\
    \le&\frac{1}{MW}\sum_{(m,w)}e(m,w,\vs)\one_\cG+\frac{1}{MW}\sum_{(m,w)\colon \cG^c\holds}e(m,w,\vs)\notag\\
    \le&\frac{1}{MW}\card{\curbrkt{(m,w)\colon I(\bfu,\bfx,\bfy;\bfs)\ge\eps}}\label{eqn:first_term}\\
    & +\frac{1}{MW}\sum_{(m,w)\colon\cG^c \holds}\sum_{\vz\colon\tau_{\vu,\vx_m,\vy_w,\vs,\vz}\notin\cP_\eta}W(\vz|\vx_m,\vy_w,\vs)\label{eqn:second_term} \\
    &+\frac{1}{MW}\sum_{(m,w)\colon\cG^c \holds}\sum_{\tau\in\cD}e_\tau(m,w,\vs)\label{eqn:third_term}\\
    &+\frac{1}{MW}\sum_{(m,w)\colon\cG^c \holds}\sum_{\tau\in\cD'}e_\tau'(m,w,\vs)\label{eqn:third_term_usr1}\\
    &+\frac{1}{MW}\sum_{(m,w)\colon\cG^c \holds}\sum_{\tau\in\cD''}e_\tau''(m,w,\vs),\label{eqn:third_term_usr2}
\end{align}
where
\begin{align}
    e_\tau(m,w,\vs)\coloneqq&\sum_{\substack{\vz\colon\tau_{\vu,\vx_m,\vy_w,\vx^{\cS\setminus m},\vy^{\cT\setminus w},\vs,\vz}=\tau\\\text{for some }\cL\in\sL_L(m,w)}}W(\vz|\vx_m,\vy_w,\vs),\notag\\
    e_\tau'(m,w,\vs) \coloneqq &\sum_{\substack{\vz\colon\tau_{\vu,\vx_m,\vy_w,\vx^{\cS\setminus m},\vs,\vz} = \tau\\\text{for some }\cS\in\binom{[M]}{L},\;\cS\ni m}}W(\vz|\vx_m,\vy_w,\vs),\notag\\
    e_\tau''(m,w,\vs)\coloneqq&\sum_{\substack{\vz\colon\tau_{\vu,\vx_m,\vy_w,\vy^{\cT\setminus w},\vs,\vz} = \tau\\\text{for some }\cT\in\binom{[W]}{L},\;\cT\ni w}}W(\vz|\vx_m,\vy_w,\vs).\notag
\end{align}
Here $\cL$ has the same underlying graph structure as $\cB_\tau$. 

The first term \eqref{eqn:first_term} is at most $2^{-n\eps/2}$ (up to polynomial factors) by Eqn. \eqref{eqn:cw_selection_xy_s}.

By Sanov's theorem, the second term \eqref{eqn:second_term} dot equals
\begin{align}
    & \sup_{P_{\bfu,\bfx,\bfy,\bfs,\bfz}\colon \cG_\atyp\cap\cG ^c\text{ holds}} 2^{-D\paren{P_{\bfu,\bfx,\bfy,\bfs,\bfz}\|P_{\bfu,\bfx,\bfy,\bfs,}W_{\bfz|\bfx,\bfy,\bfs}}}\le2^{-n(\eta-\eps)}\notag
\end{align}
where the inequality follows since 
\begin{align}
    D\paren{P_{\bfu,\bfx,\bfy,\bfs,\bfz}\|P_{\bfu,\bfx,\bfy,\bfs,}W_{\bfz|\bfx,\bfy,\bfs}} = &D(P_{\bfu,\bfx,\bfy,\bfs,\bfz}\|P_{\bfu}P_{\bfx|\bfu}P_{\bfy|\bfu}P_\bfs W_{\bfz|\bfx,\bfy,\bfs}) - I(\bfx,\bfy;\bfs|\bfu)\notag\\
    >&\eta-\eps.\notag
\end{align}

We now proceed bounding  \eqref{eqn:third_term}, \eqref{eqn:third_term_usr1} and \eqref{eqn:third_term_usr2} separately.

\subsubsection{Bounds on term \eqref{eqn:third_term}}
We can further bound term \eqref{eqn:third_term} as
\begin{align}
    \frac{1}{MW}\sum_{(m,w)\colon\cG^c \holds}\sum_{\tau\in\cD}e_\tau(m,w,\vs)\le&\frac{1}{MW}\sum_{(m,w)\colon\cG^c \holds}\sum_{\substack{\cB\in\sB_L\\I>1,J>1}}\sum_{\tau\in\cD(\cB)}e_\tau(m,w,\vs).\label{eqn:third_term_sumrate}
\end{align}
Fix any bipartite graph $\cB\in\sB_L$ with $I>1,J>1$. Given $\tau\in\cD(\cB)$ with underlying graph structure $\cB$, let
\begin{align}
    \cH_1\coloneqq&\curbrkt{R_1+R_2<\min\curbrkt{I(\bfx_i,\bfy_j;\bfs|\bfu)\colon(i,j)\in\cE(\cB)}}.\notag
\end{align}
Then
\begin{align}
    \frac{1}{MW}\sum_{\tau\in\cD(\cB)\colon\cG^c \holds}\sum_{(m,w)}e_\tau(m,w,\vs)
    =&\frac{1}{MW}\sum_{\tau\in\cD(\cB)\colon\cG^c \holds}\sum_{(m,w)}e_\tau(m,w,\vs)\one_{\cH_1}\label{eqn:third_a}\\
    &+\frac{1}{MW}\sum_{\tau\in\cD(\cB)\colon\cG^c \holds}\sum_{(m,w)}e_\tau(m,w,\vs)\one_{\cH_1^c}.\label{eqn:third-b}
\end{align}
We will bound term \eqref{eqn:third_a} and \eqref{eqn:third-b} separately and hence obtain a bound on term \eqref{eqn:third_term_sumrate}.

Let
\begin{align}
    \cH_2\coloneqq&\curbrkt{I(\bfx,\bfy;\bfx^{I-1},\bfy^{J-1},\bfs|\bfu)\ge\eps}.\notag
\end{align}
Then the term \eqref{eqn:third_a}  can be decomposed as
\begin{align}
    &\frac{1}{MW}\sum_{\tau\in\cD(\cB)\colon\cG^c\cap\cH_1\cap\cH_2 \holds}\sum_{(m,w)}e_\tau(m,w,\vs)\label{eqn:third-a-i}\\
    &+\frac{1}{MW}\sum_{\tau\in\cD(\cB)\colon\colon\cG^c \cap\cH_1\cap\cH_2^c\holds}\sum_{(m,w)}e_\tau(m,w,\vs)\label{eqn:third-a-ii}.
\end{align}
The term \eqref{eqn:third-a-i} is at most $2^{-n\eps/2}$ (up to polynomial factors) by Eqn. \eqref{eqn:cw_selection_xy_slistxy}. The term \eqref{eqn:third-a-ii} is at most
\begin{align}
    &\frac{1}{MW}\sum_{\tau\in\cD(\cB)\colon\cG^c\cap\cH_1\cap\cH_2^c\holds}\sum_{(m,w)}e_\tau(m,w,\vs)\notag\\
    \le&\frac{1}{MW}\sum_{\tau\in\cD(\cB)\colon\cG^c\cap\cH_1\cap\cH_2^c\holds}\sum_{(m,w)}\sum_{\substack{\cL\in\sL_L(m,w)\\\tau_{\vu,\vx_m,\vy_w,\vx^{\cS\setminus m},\vy^{\cT\setminus w},\vs}=[\tau]_{\bfu,\bfx,\bfy,\bfx^{I-1},\bfy^{J-1},\bfs}}}\sum_{\vz\colon\tau_{\vu,\vx_m,\vy_w,\vx^{\cS\setminus m},\vy^{\cT\setminus w},\vs,\vz}=\tau}W(\vz|\vx_m,\vy_w,\vs)\notag
\end{align}
The inner sum is at most
\begin{align}
    \sum_{\vz\colon\tau_{\vu,\vx_m,\vy_w,\vx^{\cS\setminus m},\vy^{\cT\setminus w},\vs,\vz}=\tau}W(\vz|\vx_m,\vy_w,\vs)&\le 2^{-nI\paren{\bfz;\bfx^{I-1},\bfy^{J-1}|\bfu,\bfx,\bfy,\bfs} }.\notag
\end{align}
Note that
\begin{align}
    I\paren{\bfz;\bfx^{I-1},\bfy^{J-1}|\bfu,\bfx,\bfy,\bfs} = & I\paren{\bfx,\bfy,\bfz;\bfx^{I-1},\bfy^{J-1}|\bfu,\bfs} - I\paren{\bfx,\bfy;\bfx^{I-1},\bfy^{J-1}|\bfu,\bfs}\notag\\
    >&\eta-I\paren{\bfx,\bfy;\bfx^{I-1},\bfy^{J-1},\bfs|\bfu}\label{eqn:by_tau_in_d}\\
    >&\eta-\eps.\label{eqn:by_h2c}
\end{align}
Eqn. \eqref{eqn:by_tau_in_d} is by $\tau\in\cD(\cB)$ and Eqn. \eqref{eqn:by_h2c} is by $\cH_2^c$.
Combining it with  \eqref{eqn:number_of_bipgh}, we have that term \eqref{eqn:third-a-ii} is at most
$
    2^{-n(\eta-2\eps)}
$. Hence \eqref{eqn:third_a} is at most the sum of Eqn. \eqref{eqn:third-a-i} and Eqn. \eqref{eqn:third-a-ii} which is in turn at most  (up to polynomial factors)
\begin{align}
    2^{-n\eps/2} + 2^{-n(\eta-2\eps)}\le&2^{-n\eps/2} +2^{-n\eps/2} = 2\cdot2^{-n\eps/2}.
\end{align}

We now bound the term \eqref{eqn:third-b}. Let us fix any $(i,j)\in\cE(\cB)$ such that $R_1+R_2\ge I(\bfx_i,\bfy_j;\bfs|\bfu)$. Let
\begin{align}
    \cH_3\coloneqq&\curbrkt{I(\bfx,\bfy;\bfx_i,\bfy_j,\bfs|\bfu) - \sqrbrkt{R_1+R_2-I(\bfx_i,\bfy_j;\bfs|\bfu)}^+\ge\eps}.\notag
\end{align}
Then
\begin{align}
    e_\tau(m,w,\vs)=&\sum_{\substack{\vz\colon\tau_{\vu,\vx_m,\vy_w,\vx^{\cS\setminus m},\vy^{\cT\setminus w},\vs,\vz}=\tau\\\text{for some }\cL\in\sL_L(m,w)}}W(\vz|\vx_m,\vy_w,\vs)\notag\\
    \le&\sum_{\substack{\vz\colon\tau_{\vu,\vx_m,\vy_w,\vx_{m'},\vy_{w'},\vs,\vz} = [\tau]_{\bfu,\bfx,\bfy,\bfx_i,\bfy_j,\bfs,\bfz}\\\text{for some }m'\ne m,w'\ne w}}W(\vz|\vx_m,\vy_w,\vs).\label{eqn:e_tau_prime}
\end{align}
Let $\wt e_\tau(m,w,\vs)$ denote the RHS \eqref{eqn:e_tau_prime}. Then \eqref{eqn:third-b} can be upper bounded as follows.
\begin{align}
    &\frac{1}{MW}\sum_{\tau\in\cD(\cB)\colon\cG^c\cap\cH_1^c\cap\cH_3\holds}\sum_{(m,w)}\wt e_\tau(m,w,\vs)\label{eqn:third_b_i}\\
    &+\frac{1}{MW}\sum_{\tau\in\cD(\cB)\colon\cG^c\cap\cH_1^c\cap\cH_3^c\holds}\sum_{(m,w)}\wt e_\tau(m,w,\vs).\label{eqn:third_b_ii}
\end{align}

By Eqn. \eqref{eqn:cw_selection_takepositivepart}, Eqn. \eqref{eqn:third_b_i} is at most $2^{-n\eps/2}$ (up to polynomial factors). 

Assume also that $P_{\bfx_i} = P_{\bfx},P_{\bfy_j} = P_\bfy$. Term \eqref{eqn:third_b_ii} is at most
\begin{align}
    &\frac{1}{MW}\sum_{\tau\in\cD(\cB)\colon\cG^c\cap\cH_1^c\cap\cH_3^c\holds}\sum_{(m,w)}\sum_{\substack{m'\ne m,w'\ne w\\ \tau_{\vu,\vx_m,\vy_w,\vx_{m'},\vy_{w'},\vs} = [\tau]_{\bfu,\bfx,\bfy,\bfx_i,\bfy_j,\bfs}}}\sum_{\vz\colon\tau_{\vu,\vx_{m},\vy_w,\vx_{m'},\vy_{w'},\vs,\vz} = \tau}W(\vz|\vx_m,\vy_w,\vs)\notag\\
    \le&2^{n\paren{\sqrbrkt{R_1+R_2-I\paren{\bfx_i,\bfy_j;\bfx,\bfy,\bfs|\bfu}}^++\eps}}2^{-nI\paren{\bfz;\bfx_i,\bfy_j|\bfu,\bfx,\bfy,\bfs}}\label{eqn:by_cw_selection}\\
    =&2^{-n\paren{I\paren{\bfz;\bfx_i,\bfy_j|\bfu,\bfx,\bfy,\bfs} - \sqrbrkt{R_1+R_2-I\paren{\bfx_i,\bfy_j;\bfx,\bfy,\bfs|\bfu}}^+-\eps}}.\label{eqn:to_be_continued}
\end{align}
The Eqn. \eqref{eqn:by_cw_selection} is by Eqn. \eqref{eqn:cw_selection_no_mprimewprime}.
Note that $\cH_1^c\cap\cH_3^c$ implies 
\begin{align}
     R_1+R_2>&I(\bfx,\bfy;\bfx_i,\bfy_j,\bfs|\bfu)+I(\bfx_i,\bfy_j;\bfs|\bfu)-\eps\notag\\
     >&I(\bfx,\bfy;\bfx_i,\bfy_j|\bfu,\bfs)+I(\bfx_i,\bfy_j;\bfs|\bfu)-\eps\notag\\
     =&I(\bfx_i,\bfy_j;\bfx,\bfy,\bfs|\bfu)-\eps.\notag
\end{align}
Therefore
\begin{align}
     \sqrbrkt{R_1+R_2-I(\bfx_i,\bfy_j;\bfx,\bfy,\bfs|\bfu)}^+\le&R_1+R_2-I(\bfx_i,\bfy_j;\bfx,\bfy,\bfs|\bfu)+\eps.\notag
\end{align}
Continuing with Eqn. \eqref{eqn:to_be_continued}, the term \eqref{eqn:third_b_ii} is  at most
\begin{align}
     &2^{-n\paren{I(\bfz;\bfx_i,\bfy_j|\bfu,\bfx,\bfy,\bfs)-(R_1+R_2)+I(\bfx_i,\bfy_j;\bfx,\bfy,\bfs|\bfu)-2\eps}}\notag\\
     =&2^{-n\paren{I(\bfx,\bfy,\bfs,\bfz;\bfx_i,\bfy_j|\bfu)-(R_1+R_2)-2\eps}}\notag\\
     \le&2^{-n\paren{I(\bfz;\bfx_i,\bfy_j|\bfu)-(R_1+R_2)-2\eps}}\notag\\
     \le&2^{-n((I(\bfz;\bfx,\bfy|\bfu) - \delta/3) - (I(\bfz;\bfx,\bfy|\bfu) - 5\delta/6)-2\eps)}\label{eqn:by_choice_of_eta_sumrate}\\
     =&2^{-n(\delta/2-2\eps)}\notag\\
     \le&2^{-n\eps/2},\label{eqn:by_choice_of_eps_usr12}
\end{align}
where Eqn. \eqref{eqn:by_choice_of_eta_sumrate} is due to the choice of $\eta$ and $R_1,R_2$, and Eqn. \eqref{eqn:by_choice_of_eps_usr12} is due to the choice of $\eps$.

Finally the term \eqref{eqn:third-b} is bounded by the sum of term \eqref{eqn:third_b_i} and term \eqref{eqn:third_b_ii} which is in turn at most $2^{-n\eps/2}+2^{-n\eps/2} = 2\cdot2^{-n\eps/2}$.
 
\subsubsection{Bounds on term \eqref{eqn:third_term_usr1}} Term \eqref{eqn:third_term_usr1} can be bounded in a similar manner to term \eqref{eqn:third_term_sumrate}. We provide the calculations for completeness. Term \eqref{eqn:third_term_usr2} is symmetric to term \eqref{eqn:third_term_usr1} and we omit the details.

Given $\tau\in\cD'$, let
\begin{align}
    \cH_1'\coloneqq&\curbrkt{R_1<\min\curbrkt{I(\bfx_i;\bfs|\bfu)\colon i\in[L]}}.\notag
\end{align}

Then
\begin{align}
    \frac{1}{MW}\sum_{\tau\in\cD'\colon\cG^c \holds}\sum_{(m,w)}e_\tau'(m,w,\vs)
    =&\frac{1}{MW}\sum_{\tau\in\cD'\colon\cG^c \holds}\sum_{(m,w)}e_\tau'(m,w,\vs)\one_{\cH_1'}\label{eqn:third-a-usr1}\\
    &+\frac{1}{MW}\sum_{\tau\in\cD'\colon\cG^c \holds}\sum_{(m,w)}e_\tau'(m,w,\vs)\one_{\cH_1'^c}.\label{eqn:third-b-usr1}
\end{align}
We will bound term \eqref{eqn:third-a-usr1} and \eqref{eqn:third-b-usr1} separately and hence obtain a bound on term \eqref{eqn:third_term_usr1}.

Let
\begin{align}
    \cH_2'\coloneqq&\curbrkt{I(\bfx,\bfy;\bfx^{L-1},\bfs|\bfu)\ge\eps}.\notag
\end{align}
Then the term \eqref{eqn:third-a-usr1}  can be decomposed as
\begin{align}
    &\frac{1}{MW}\sum_{\tau\in\cD'\colon\cG^c\cap\cH_1'\cap\cH_2' \holds}\sum_{(m,w)}e_\tau'(m,w,\vs)\label{eqn:third-a-i-usr1}\\
    &+\frac{1}{MW}\sum_{\tau\in\cD'\colon\colon\cG^c \cap\cH_1'\cap\cH_2'^c\holds}\sum_{(m,w)}e_\tau'(m,w,\vs)\label{eqn:third-a-ii-usr1}.
\end{align}
The term \eqref{eqn:third-a-i-usr1} is at most $2^{-n\eps/2}$ (up to polynomial factors) by Eqn. \eqref{eqn:cw_selection_xy_slistx}. The term \eqref{eqn:third-a-ii-usr1} is at most
\begin{align}
    &\frac{1}{MW}\sum_{\tau\in\cD'\colon\cG^c\cap\cH_1'\cap\cH_2'^c\holds}\sum_{(m,w)}e_\tau'(m,w,\vs)\notag\\
    \le&\frac{1}{MW}\sum_{\tau\in\cD'\colon\cG^c\cap\cH_1'\cap\cH_2'^c\holds}\sum_{(m,w)}\sum_{\substack{\cS\in\binom{[M]}{L},\;\cS\ni m\\\tau_{\vu,\vx_m,\vy_w,\vx^{\cS\setminus m},\vs}=[\tau]_{\bfu,\bfx,\bfy,\bfx^{L-1},\bfs}}}\sum_{\vz\colon\tau_{\vu,\vx_m,\vy_w,\vx^{\cS\setminus m},\vs,\vz}=\tau}W(\vz|\vx_m,\vy_w,\vs)\notag
\end{align}
The inner sum is at most
\begin{align}
    \sum_{\vz\colon\tau_{\vu,\vx_m,\vy_w,\vx^{\cS\setminus m},\vs,\vz}=\tau}W(\vz|\vx_m,\vy_w,\vs)&\le 2^{-nI\paren{\bfz;\bfx^{L-1}|\bfu,\bfx,\bfy,\bfs} }.\notag
\end{align}
Note that
\begin{align}
    I\paren{\bfz;\bfx^{L-1}|\bfu,\bfx,\bfy,\bfs} = & I\paren{\bfx,\bfy,\bfz;\bfx^{L-1}|\bfu,\bfs} - I\paren{\bfx,\bfy;\bfx^{L-1}|\bfu,\bfs}\notag\\
    >&\eta-I\paren{\bfx,\bfy;\bfx^{L-1},\bfs|\bfu}\label{eqn:by_tau_in_d_usr1}\\
    >&\eta-\eps.\label{eqn:by_h2c_usr1}
\end{align}
Eqn. \eqref{eqn:by_tau_in_d_usr1} is by $\tau\in\cD'$ and Eqn. \eqref{eqn:by_h2c_usr1} is by $\cH_2'^c$.
Combining it with  \eqref{eqn:number_of_s}, we have that term \eqref{eqn:third-a-ii-usr1} is at most
$
    2^{-n(\eta-2\eps)}
$. Hence \eqref{eqn:third-a-usr1} is at most the sum of Eqn. \eqref{eqn:third-a-i-usr1} and Eqn. \eqref{eqn:third-a-ii-usr1}, which is in turn at most $2^{-n\eps/2} + 2^{-n(\eta-2\eps)}$ (up to polynomial factors). 

We now bound the term \eqref{eqn:third-b-usr1}. Let us fix any $i\in[L]$ such that $R_1\ge I(\bfx_i;\bfs|\bfu)$. Let
\begin{align}
    \cH_3'\coloneqq&\curbrkt{I(\bfx,\bfy;\bfx_i,\bfs|\bfu) - \sqrbrkt{R_1-I(\bfx_i;\bfs|\bfu)}^+\ge\eps}.\notag
\end{align}
Then
\begin{align}
    e_\tau'(m,w,\vs)=&\sum_{\substack{\vz\colon\tau_{\vu,\vx_m,\vy_w,\vx^{\cS\setminus m},\vs,\vz}=\tau\\\text{for some }\cS\in\binom{[M]}{L},\;\cS\ni m}}W(\vz|\vx_m,\vy_w,\vs)\notag\\
    \le&\sum_{\substack{\vz\colon\tau_{\vu,\vx_m,\vy_w,\vx_{m'},\vs,\vz} = [\tau]_{\bfu,\bfx,\bfy,\bfx_i,\bfs,\bfz}\\\text{for some }m'\ne m}}W(\vz|\vx_m,\vy_w,\vs).\label{eqn:e_tau_prime_usr1}
\end{align}
Let $\wt e_\tau'(m,w,\vs)$ denote the RHS \eqref{eqn:e_tau_prime_usr1}. Then \eqref{eqn:third-b-usr1} can be upper bounded as follows.
\begin{align}
    &\frac{1}{MW}\sum_{\tau\in\cD'\colon\cG^c\cap\cH_1'^c\cap\cH_3'\holds}\sum_{(m,w)}\wt e_\tau'(m,w,\vs)\label{eqn:third-b-i-usr1}\\
    &+\frac{1}{MW}\sum_{\tau\in\cD'\colon\cG^c\cap\cH_1'^c\cap\cH_3'^c\holds}\sum_{(m,w)}\wt e_\tau'(m,w,\vs).\label{eqn:third-b-ii-usr1}
\end{align}

By Eqn. \eqref{eqn:cw_selection_takepositivepart_usr1}, Eqn. \eqref{eqn:third-b-i-usr1} is at most $2^{-n\eps/2}$ (up to polynomial factors). 

Assume also that $P_{\bfx_i} = P_{\bfx},P_{\bfy_j} = P_\bfy$. Term \eqref{eqn:third-b-ii-usr1} is at most
\begin{align}
    &\frac{1}{MW}\sum_{\tau\in\cD'\colon\cG^c\cap\cH_1'^c\cap\cH_3'^c\holds}\sum_{(m,w)}\sum_{\substack{m'\ne m\\ \tau_{\vu,\vx_m,\vy_w,\vx_{m'},\vs} = [\tau]_{\bfu,\bfx,\bfy,\bfx_i,\bfs}}}\sum_{\vz\colon\tau_{\vu,\vx_{m},\vy_w,\vx_{m'},\vs,\vz} = \tau}W(\vz|\vx_m,\vy_w,\vs)\notag\\
    \le&2^{n\paren{\sqrbrkt{R_1-I\paren{\bfx_i;\bfx,\bfy,\bfs|\bfu}}^++\eps}}2^{-nI\paren{\bfz;\bfx_i|\bfu,\bfx,\bfy,\bfs}}\label{eqn:by_cw_selection_usr1}\\
    =&2^{-n\paren{I\paren{\bfz;\bfx_i|\bfu,\bfx,\bfy,\bfs} - \sqrbrkt{R_1-I\paren{\bfx_i;\bfx,\bfy,\bfs|\bfu}}^+-\eps}}.\label{eqn:to_be_continued_usr1}
\end{align}
The Eqn. \eqref{eqn:by_cw_selection_usr1} is by Eqn. \eqref{eqn:cw_selection_no_mprime}.
Note that $\cH_1'^c\cap\cH_3'^c$ implies 
 \begin{align}
     R_1>&I(\bfx,\bfy;\bfx_i,\bfs|\bfu)+I(\bfx_i;\bfs|\bfu)-\eps\notag\\
     >&I(\bfx,\bfy;\bfx_i|\bfu,\bfs)+I(\bfx_i;\bfs|\bfu)-\eps\notag\\
     =&I(\bfx_i;\bfx,\bfy,\bfs|\bfu)-\eps.\notag
 \end{align}
 Therefore
 \begin{align}
     \sqrbrkt{R_1-I(\bfx_i;\bfx,\bfy,\bfs|\bfu)}^+\le&R_1-I(\bfx_i;\bfx,\bfy,\bfs|\bfu)+\eps.\notag
 \end{align}
Continuing with Eqn. \eqref{eqn:to_be_continued_usr1}, the term \eqref{eqn:third-b-ii-usr1} is  at most
\begin{align}
    &2^{-n\paren{I(\bfz;\bfx_i|\bfu,\bfx,\bfy,\bfs)-R_1+I(\bfx_i;\bfx,\bfy,\bfs|\bfu)-2\eps}}\notag\\
    =&2^{-n\paren{I(\bfx,\bfy,\bfs,\bfz;\bfx_i|\bfu)-R_1-2\eps}}\notag\\
    \le&2^{-n\paren{I(\bfz;\bfx_i|\bfu)-R_1-2\eps}}\notag\\
    \le&2^{-n((I(\bfz;\bfx|\bfu) - \delta/3) - (I(\bfz;\bfx|\bfu) - 5\delta/6)-2\eps)}\\
     =&2^{-n(\delta/2-2\eps)}\notag\\
     \le&2^{-n\eps/2}.
\end{align}

Finally the term \eqref{eqn:third-b-usr1} is bounded by the sum of term \eqref{eqn:third-b-i-usr1} and term \eqref{eqn:third-b-ii-usr1} which is in turn at most $2^{-n\eps/2}+2^{-n\eps/2} = 2\cdot2^{-n\eps/2}$.


\subsection{Converse} Assume $L\le L_s$. Without loss of generality, it suffices to set $L=L_s$. We want to show that any code has strictly positive average probability of error.

By non-symmetrizability, there exist a bipartite graph $\cB = (\cI,\cJ,\cE)\in\sB_{L}$, distributions $P_\bfu,P_{\bfx,\bfy|\bfu}$, and a collection of symmetrizing distributions $\curbrkt{Q_{\bfs|\bfx^{I-1},\bfy^{J-1}}^{(u)}}_u\subset\cQ_\symm(\cB)$ such that for any $P_{\bfx^{I-1},\bfy^{J-1}|\bfu}$ with  $\sqrbrkt{P_{\bfx^{I-1},\bfy^{J-1}|\bfu}}_{\bfx_i,\bfy_j|\bfu} = P_{\bfx,\bfy|\bfu}$, 
\begin{align}
    \overline{\Lambda}_s(\cB, P_{\bfu,\bfx^{I-1},\bfy^{J-1}})\coloneqq&\sum_{u,x^{I-1},y^{J-1},s}P(u)P(x^{I-1},y^{J-1}|u)Q^{(u)}(s|x^{I-1},y^{J-1})g(s)\notag\\
    <&\Lambda,\notag
\end{align}
where $P_{\bfu,\bfx^{I-1},\bfy^{J-1}} = P_\bfu P_{\bfx^{I-1},\bfy^{J-1}|\bfu}$. Let us assume $\overline{\Lambda}_s(\cB,P_{\bfu,\bfx^{I-1},\bfy^{J-1}}) = \Lambda-\delta$ for some constant $\delta>0$.

Consider the following jamming strategy of James. Fix any $\vu$ of type $P_\bfu$. Sample $\cS\sim\binom{[M]}{I-1}$ and $\cT\sim\binom{[W]}{J-1}$ independently, uniformly at random. Generate $\vbfs$ according to the distribution 
\begin{align}
    U_{\vbfs|\vx^\cS,\vy^\cT}(\vs|\vu,\vx^{\cS},\vy^\cT) = &\prod_u\prod_{\substack{i\in[n]\\\vu(i)=u}}Q^{(u)}(\vs(i)|\vx^\cS(i),\vy^\cT(i)).\notag
\end{align}
If $g(\vbfs)>\Lambda$, then transmit a fixed vector $\wt\vs =[s_0,\cdots,s_0]$ where $s_0=\argmax{s}g(s)$. Therefore $g(\wt\vs) = g_{\min}$ where $g_{\min}=\min_sg(s)=g(s_0)$. The jamming vector transmitted by James satisfies the state constraint with probability 1.

Given $\cS,\cT$, the expected cost of $\vbfs$  is
\begin{align}
    \exptover{\vbfs\sim U}{g(\vbfs)|\cS,\cT}=&\frac{1}{n}\sum_{i=1}^n\expt{g(\vbfs(i)|\cS,\cT)}\notag\\
    =&\frac{1}{n}\sum_{i}\sum_s U(s|\vx^{\cS}(i),\vy^{\cT}(i))g(s)\notag\\
    =&\frac{1}{n}\sum_s\sum_{u,x^\cS,y^\cT}\sum_{\substack{i\in[n]\\\vu(i)=u\\\vx^\cS(i)=x^\cS\\\vy^\cT(i)=y^\cT}}g(s)Q^{(u)}(s|x^\cS,y^\cT)\notag\\
    =&\sum_{u,x^{\cS},y^\cT,s}Q^{(u)}(s|x^\cS,y^\cT)g(s)\frac{\card{\curbrkt{i\in[n]\colon \begin{array}{l}
         \vu(i)=u,\\
         \vx^\cS(i) = x^\cS,\\
         \vy^\cT(i) = y^\cT  
    \end{array} }}}{n}\notag\\
    =&\sum_{u,x^{I-1},y^{J-1},s}Q^{(u)}(s|x^\cS,y^\cT)g(s)P(u,x^{I-1},y^{J-1})\notag\\
    =&\overline{\Lambda}_s(\cB,P_{\bfu,\bfx^{I-1},\bfy^{J-1}})\notag\\
    =&\Lambda-\delta.\notag
\end{align}
We can also bound the variance (conditioned on $\cS,\cT$) of the cost of James' jamming vector.
\begin{align}
    \varover{\vbfs\sim U}{g(\vbfs)|\cS,\cT} = &\var{\left. \frac{1}{n}\sum_ig(\vbfs(i))\right|\cS,\cT }\notag\\
    =&\frac{1}{n^2}\sum_i\var{g(\vbfs(i))|\cS,\cT}\label{eqn:s_indep_coord}\\
    \le&\frac{1}{n^2}\sum_i\expt{g^2(\vbfs(i))|\cS,\cT}\notag\\
    \le&\frac{1}{n^2}ng_*^2\notag\\
    =&{g_*^2}/{n},\notag
\end{align}
where $g_*\coloneqq\max_sg(s)$. Eqn. \eqref{eqn:s_indep_coord} follows since each coordinate of $\vbfs$ is independent. Now,
\begin{align}
    \prob{g(\vbfs)>\Lambda|\cS,\cT}\le&\prob{|g(\vbfs) - \expt{g(\vbfs)}|>\Lambda - \expt{g(\vbfs)}|\cS,\cT}\label{eqn:to_apply_cheb}
\end{align}
Note that
\begin{align}
    \Lambda-\expt{g(\vbfs)|\cS,\cT} = &\overline{\Lambda}_s(\cB,P_{\bfu,\bfx^{I-1},\bfy^{J-1}})-\Lambda\notag\\
    =&\delta>0.\notag
\end{align}
Hence by Chebyshev's inequality, the probability \eqref{eqn:to_apply_cheb} is at most
\begin{align}
    \frac{\var{g(\vbfs)|\cS,\cT}}{\paren{\Lambda - \expt{g(\vbfs)|\cS,\cT}}^2}
    \le&\frac{g_*^2/n}{\delta^2}
    =\frac{g_*^2}{n\delta^2}.\notag
\end{align}

Note that for any $(m,w)\notin\cS\times\cT$ and any $(m',w')\in\cS\times\cT$, we have
\begin{align}
    \exptover{\vbfs\sim U}{W^\tn(\vz|\vx_m,\vy_w,\vbfs)|\cS,\cT} 
    =& \prod_i\expt{\left.W(\vz(i)|\vx_m(i),\vy_w(i),\vbfs(i))\right|\cS,\cT}\label{eqn:first_eq}\\
    =&\prod_i\sum_sW(\vz(i)|\vx_m(i),\vy_w(i),s)U(s|\vx^{\cS}(i),\vy^\cT(i))\notag\\
    =&\prod_u\prod_{\substack{i\in[n]\\\vu(i)=u}}\sum_sW(\vz(i)|\vx_m(i),\vy_w(i),s)Q^{(u)}(s|\vx^\cS(i),\vy^\cT(i))\notag\\
    =&\prod_u\prod_{\substack{i\in[n]\\\vu(i)=u}}\sum_sW(\vz(i)|\vx_{m'}(i),\vy_{w'}(i),s)Q^{(u)}\paren{s\left|\vx^{(\cS\setminus m')\cup m}(i),\vy^{(\cT\setminus w')\cup w}(i)\right.}\label{eqn:identity_by_symm}\\
    =&\expt{\left.W^\tn(\vz|\vx_{m'},\vy_{w'},\vbfs)\right|(\cS\setminus m')\cup m,(\cT\setminus w')\cup w},\label{eqn:roll_back}
\end{align}
where Eqn. \eqref{eqn:identity_by_symm} follows since  $\curbrkt{Q^{(u)}}_u$ are symmetrizing distributions. Eqn. \eqref{eqn:roll_back} follows by rolling equalities \eqref{eqn:first_eq} to \eqref{eqn:identity_by_symm} back.
We thus have that, for any $\cL = (\cS',\cT',\cF')\in\sL_{L+1}(m,w)$ with the same underlying graph structure as $\cB$ and for some $(m_0,w_0)\in\cF'$,
\begin{align}
    \sum_{(m,w)\in\cF'}\exptover{\vbfs}{e(m,w,\vbfs)|\cS'\setminus m,\cT'\setminus w} = &\sum_{(m,w)\in\cF'}\paren{1 - \sum_{\vz\colon(m,w)\in\psi(\vz)}\expt{W(\vz|\vx_m,\vy_w,\vbfs)|\cS'\setminus m,\cT'\setminus w}}\notag\\
    =&(L+1)-\sum_{\substack{(m,w)\in\cF'\\ (m,w)\in\psi(\vz)}}\sum_{\vz}\expt{W(\vz|\vx_{m_0},\vy_{w_0},\vs)|\cS'\setminus m_0,\cT'\setminus w_0}\notag\\
    \ge&(L+1) - L\label{eqn:ineq_by_dec}\\
    =&1,\label{eqn:obs_to_use}
\end{align}
where inequality \eqref{eqn:ineq_by_dec} follows since, by list decodability requirement, $|\psi(\vz)|\le L$ for any received $\vz$.
Using the above observations, the expected (over jamming strategy) average probability of error can be lower bounded as follows.
\begin{align}
    \exptover{\cS,\cT,\vbfs}{P_{e,\avg}(\vbfs)} = &\frac{1}{\binom{M}{I-1}}\frac{1}{\binom{W}{J-1}}\frac{1}{MW}\sum_{\cS,\cT}\sum_{(m,w)}\exptover{\vbfs}{e(m,w,\vbfs)|\cS,\cT}\notag\\
    \ge&\frac{1}{\binom{M}{I-1}}\frac{1}{\binom{W}{J-1}}\frac{1}{MW}\sum_{\cS',\cT'}\sum_{(m,w)\in\cF'}\expt{e(m,w,\vbfs)|\cS'\setminus m,\cT'\setminus w}\label{eqn:by_obs}\\
    \ge&\frac{1}{\binom{M}{I-1}}\frac{1}{\binom{W}{J-1}}\frac{1}{MW}\binom{M}{I}\binom{W}{J}\notag\\
    =&\paren{\frac{1}{I}-\frac{1}{M}+\frac{1}{MI}}\paren{\frac{1}{J} - \frac{1}{W}+\frac{1}{WJ}},\notag
\end{align}
where inequality \eqref{eqn:by_obs} follows since the inner summation is at least 1 by inequality \eqref{eqn:obs_to_use}. Since the above bound holds averaged over James' stochastic jamming strategy, there exists $\vs$ generated from $U$ \footnote{Note again that $\vs$ sampled from $U$ must satisfy state constraints.} such that
\begin{align}
    P_{e,\avg}\ge &P_{e,\avg}(\vs)\notag\\
    \ge&\paren{\frac{1}{I}-\frac{1}{M}+\frac{1}{MI}}\paren{\frac{1}{J} - \frac{1}{W}+\frac{1}{WJ}}.\notag
\end{align}
For any codebook pair of positive rate $R_1,R_2$, $M\xrightarrow{n\to\infty}\infty,W\xrightarrow{n\to\infty}\infty$. Noe that for a bipartite graph $\cB\in\sB_L$ with $I$ left vertices and $J$ right vertices to have no isolated vertex, $L=|\cE(\cB)|\ge\max\curbrkt{I,J}$ Hence
$P_{e,\avg}\xrightarrow{n\to\infty}\frac{1}{IJ}\ge1/L^2$.

\subsection{Recovering \cite{cai-2016-list-dec-obli-avmac}}
Our results recovers the list decoding results for \emph{unconstrained} AVMACs by Cai \cite{cai-2016-list-dec-obli-avmac}. 
By setting 
\begin{align}
\Gamma_1 \coloneqq& \max_{x\in\cX} f_1(x), \quad
\Gamma_2 \coloneqq \max_{y\in\cY} f_2(y), \notag \\
\Lambda \coloneqq& \max_{s\in\cS} g(s), \notag
\end{align}
we have, for every $ P_{\bfu,\bfx,\bfy} $,
\begin{align}
L_{s}(P_{\bfu,\bfx,\bfy}) = L_{w}(P_{\bfu,\bfx,\bfy}) = L_{\symm} \coloneqq& \max\curbrkt{ L\colon \exists \cB = (\cI,\cJ,\cE)\text{ s.t. }  |\cE| = L,\; \cQ_{\symm}(\cB)\ne\emptyset }. \notag
\end{align}
our inner and outer bounds collapse to the same region which matches Cai's characterization stated below.
If $ L>L_{\symm} $, then the $L$-list decoding capacity of unconstrained oblivious AVMAC is given by
\begin{align}
C =& \bigcup_{P_{\bfu,\bfx,\bfy} = P_{\bfu}P_{\bfx|\bfu}P_{\bfy|\bfu}} \curbrkt{(R_1,R_2)\colon \begin{array}{rl}
R_1 \le& \inf I(\bfx;\bfz|\bfy), \\
R_2\le& \inf I(\bfy;\bfz|\bfx), \\
R_1+R_2\le& \inf I(\bfx,\bfy;\bfz)
\end{array} }, \notag
\end{align}
where the infimum is taken over jamming distribution $ P_{\bfs|\bfu}\in\Delta(\cS|\cU) $ and the mutual information is evaluated w.r.t. distribution $ P_{\bfu}P_{\bfx|\bfu}P_{\bfy|\bfu}P_{\bfs|\bfu}W_{\bfz|\bfx,\bfy,\bfs} $.

If $ L\le L_{\symm} $ then the capacity region has empty interior. 

\begin{remark}
Cai's results were originally stated in terms of closure of convex hulls of multiple regions. 
Here we adopt an equivalent formulation by introducing a time-sharing variable $ \bfu $.
\end{remark}

\section{List decoding Gaussian AVMACs}\label{sec:gaussian_avmac}
\subsection{Model} Suppose user one has $M \coloneqq  L2^{nR_1}$ messages and user two has $W \coloneqq  L2^{nR_2}$ messages. They both have access to a MAC which is governed by an adversary. To transmit a message pair $(m,w)\in[M]\times[W]$ which is uniformly distributed, two users (who are not allowed to cooperate)  encode their messages to  length-$n$ real-valued codewords $\vx$ and $\vy$, respectively, subject to the input power constraints
\begin{align}
\normtwo{\vx}\le\sqrt{nP_1},\quad\normtwo{\vy}\le\sqrt{nP_2}. \notag
\end{align}
The adversary can introduce adversarial noise $\vs$ subject to the state power constraint
\begin{align}
\normtwo{\vs}\le& \sqrt{nN}\notag
\end{align}
only based on his knowledge of two users' codebooks. The channels add up $\vx$, $\vy$, $\vs$ together with a Gaussian noise $\vg$ whose components are i.i.d. Gaussians of variance $\sigma^2$. That is, the channel  outputs $\vz = \vx+\vy+\vs+\vg$. The receiver is required to estimate $(m,w)$ given the received vector $\vz$. 
See Fig. \ref{fig:list_dec_obli-gaussian-avmac} for the system diagram of list decoding for oblivious Gaussian AVMACs.
\begin{figure}[htbp]
  \centering
  \includegraphics[width=0.95\textwidth]{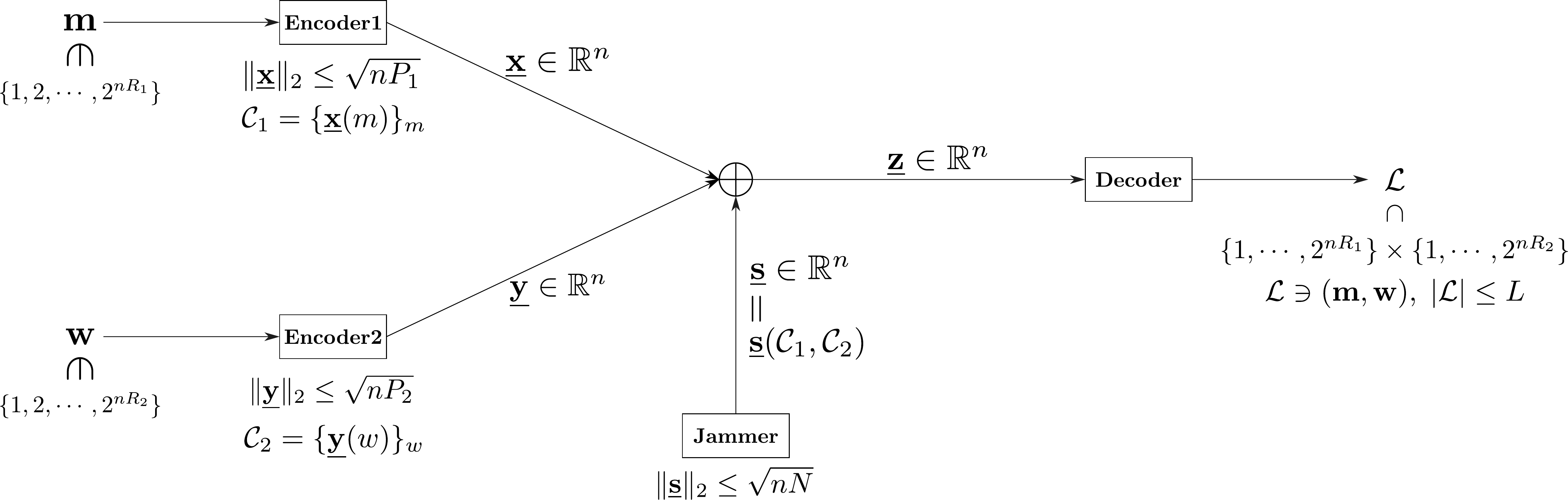}
  \caption{List decoding for oblivious Gaussian AVMACs.}
  \label{fig:list_dec_obli-gaussian-avmac}
\end{figure}

\subsection{Code design} 
The codebook $\cC_1 = \curbrkt{\vbfx_m}_{m\in[M]}$ of user one consists of $M$ i.i.d. random vectors uniformly distributed on the $(n-1)$-dimensional sphere $\cS^{n-1}(\vzero,\sqrt{nP_1})$ of radius $\sqrt{nP_1}$. The codebook $\cC_2 = \curbrkt{\vbfy_w}_{w\in[W]}$ of user two consists of $W$ i.i.d. random vectors uniformly distributed on $\cS^{n-1}(\vzero,\sqrt{nP_2})$.

Let $\phi_1$ and $\phi_2$ denote the encoding function of user one and two respectively. Let $\psi$  denote the decoding function of the receiver. 


The receiver uses a minimum distance decoder. Given the received vector $\vz$, the decoder finds codeword pairs $$(\vx_{m_1},\vy_{w_1}),\cdots,(\vx_{m_L},\vy_{w_L})$$ which are the first $L$ closest codeword pairs to $\vz$. That is
\begin{align}
\normtwo{\vz - \paren{\vx_{m_1} + \vy_{w_1}}}\le\cdots\le\normtwo{\vz - \paren{\vx_{m_L} + \vy_{w_L}}}\notag
\end{align}
and
\begin{align}
\normtwo{\vz - \paren{\vx_{m'} + \vy_{w'}}}\ge\normtwo{\vz - \paren{\vx_{m_L} + \vy_{w_L}}}\notag
\end{align}
for any other $m',w'$. The decoder then outputs $\psi(\vz) = \cL = \curbrkt{(m_1,w_1),\cdots,(m_L,w_L)}$. 

\subsection{Achievability} Along the lines of \cite{csiszar-narayan-1991-gavc}, it can be shown that  whenever $LP_1>N$ and $LP_2>N$, any rate pair $(R_1,R_2)$ satisfying
\begin{align}
R_1<&\frac{1}{2}\log\paren{1+\frac{P_1}{N+\sigma^2}},\label{eqn:gavc_r1}\\
R_2<&\frac{1}{2}\log\paren{1+\frac{P_2}{N+\sigma^2}},\label{eqn:gavc_r2}\\
R_1+R_2<&\frac{1}{2}\log\paren{1+\frac{P_1+P_2}{N+\sigma^2}}\label{eqn:gavc_r1r2}
\end{align}
can be achieved.

\subsection{Converse} 
When $ LP_1>N $ and $ LP_2>N $, the outer bound follows again from strong converse to list decoding (non-adversarial) Gaussian MACs, whose details we omit.

Our converse in the zero-rate regime pursues the geometric approach instead of reducing it to the discrete alphabet case by quantization. 
The argument is inspired by a novel bounding trick introduced in a recent work \cite{hosseinigoki-kosut-2018-oblivious-gaussian-avc-ld}.

When $LP_1>N$ and $LP_2>N$, by letting James transmit Gaussian noise $\vbfg'\sim\cN\paren{\vzero,(N-\eta)\bfI_n}$ for an arbitrarily small constant $\eta>0$, we can show that there exists no $L$-list decodable code of rate $(R_1,R_2)$ not satisfying Eqn. \eqref{eqn:gavc_r1}, \eqref{eqn:gavc_r2} and \eqref{eqn:gavc_r1r2}. This is because under such a jamming strategy, the channel is turned into a (non-adversarial) Gaussian MAC
$\vbfz = \vbfx+\vbfy+\vbfg''$,
where $\vbfg'' = \vbfg+\vbfg'\sim\cN\paren{\vzero,(N+\sigma^2-\eta)\bfI_n}$. The result follows from the converse for list decoding for Gaussian MACs.

We now show that when $LP_1<N$ and $LP_2<N$, no positive rate can be achieved. Suppose $LP_1(1+\delta_1)=N$ and $LP_2(1+\delta_2)=N$ for some constants $\delta_1>0,\delta_2>0$. We equip James with the following jamming strategy. Suppose a codebook pair $(\cC_1,\cC_2)$ is $L$-list decodable. To jam the communication, James first flips a fair coin. If the output is 1, then he samples $\vbfx_1,\cdots,\vbfx_L$ from $\cC_2$ uniformly and independently. He transmits $\vbfs = \vbfx_1+\cdots+\vbfx_L-L\vu$ if $\normtwo{\vbfs}\le\sqrt{nN}$ and transmits $\vzero$ otherwise. If the output is 0, then James samples   $\vbfy_1,\cdots,\vbfy_L$ uniformly and independently from $\cC_2$.  He transmits $\vbfs = \vbfy_1+\cdots+\vbfy_L-L\vv$ if $\normtwo{\vbfs}\le\sqrt{nN}$ and transmits $\vzero$ otherwise. Here $\vu$ and $\vv$ are two shift vectors that can be computed based purely on $\cC_1,\cC_2$. The construction of  $\vu,\vv$ is described below.

Let 
\begin{align}
\eta_*\coloneqq\inf\curbrkt{\eta\ge0\colon \liminf_{n\to\infty}\sup_{\vu\in\bR^n}\probover{\vbfx\sim\cC_1}{\normtwo{\vbfx - \vu}\le\sqrt{nP_1\eta}}>0}.\notag
\end{align}
Note that $\eta_*\le1$ since $\probover{\vbfx\sim\cC_1}{\normtwo{\vbfx}\le\sqrt{nP_1}} = 1$ for all $n$. 
Fix $\gamma>0$ such that
\begin{align}
\gamma \coloneqq& \min\curbrkt{ \frac{ \sqrt{ (\delta_1+2\eta_* - 4L\eta_* + 2L^2\eta_*)^2 + (2L^2-4L+1)\delta_1^2 } - (\delta_1+2\eta_*-4L\eta_*+2L^2\eta_*) }{ 2L^2 - 4L + 1 }, \eta_* }. \label{eqn:def_gamma}
\end{align}\tabularnewline
Let $\eta\coloneqq\eta_*+\gamma/2$. 
The first term on the RHS of Eqn. \eqref{eqn:def_gamma} is the unique positive root of the following equation
\[ \left\{\begin{array}{l}
\gamma/2+(L-1)\sqrt{\gamma\eta} = \delta_1/2 \\
\eta = \eta_*+\gamma/2
\end{array}.\right. \]
Define
\begin{align}
\eps\coloneqq \liminf_{n\to\infty}\sup_{\vu\in\bR^n}\probover{\vbfx\sim\cC_1}{\normtwo{\vbfx - \vu}\le\sqrt{nP_1\eta}}>0.\notag
\end{align}
Therefore, for sufficiently large $n$, there exists a $\vu\in\bR^n$ such that 
\begin{align}
\probover{\vbfx\sim\cC_1}{\normtwo{\vbfx - \vu}\le\sqrt{nP_1\eta}}\ge{\eps/2}.\label{eqn:prob_x_in_k}
\end{align}
This $\vu$ is what James uses in his jamming strategy.  $\vv$ can be found similarly.

It remains to show that under such a jamming strategy, the probability of error is non-vanishing in $n$ if the sizes of $\cC_1$ and $\cC_2$ are too large. For notational convenience, let $\vbfx_0\coloneqq\vbfx$ and $\vbfy_0\coloneqq\vbfy$. Note that if the coin flip is 1, Bob receives $\vbfz = \vbfx_0 + \vbfx_1 + \cdots + \vbfx_L - L\vu + \vbfy + \vbfg$. If for any size-$L$ subset $\cL\subset\curbrkt{0,1,\cdots,L}$, 
\[\normtwo{\sum_{i\in\cL}\vbfx_i-L\vu}\le\sqrt{nN},\]
then it is impossible for Bob to tell, among $L+1$ codewords,  which $L$-sized subset of codewords were forged by James and which one was transmitted by user one. Hence, conditioned on that $\vbfx_0,\vbfx_1,\cdots,\vbfx_L$ are distinct, even using the optimal decoder, the decoding error probability is at least $\frac{1}{L+1}$  since $\vbfx_0,\vbfx_1,\cdots,\vbfx_L$ appear indistinguishable to him. 
Even the decoder knew $\vbfy$ (the encoding of user two's message $w$) was transmitted, there is nothing that the he can do better than randomly guessing a $L$-sized subset $\cL = \curbrkt{i_1,\cdots,i_L}\subset\curbrkt{0,1,\cdots,L}$ and outputting the list
$(i_1,w),\cdots,(i_L,w)$.
Similarly, if the coin flip is 0, 
\[\normtwo{\sum_{i\in\cL}\vbfy_i-L\vv}\le\sqrt{nN},\]
for all $\cL\in\binom{\curbrkt{0,1,\cdots,L}}{L}$, and $\vbfy_0,\vbfy_1,\cdots,\vbfy_L$ are distinct, then the decoding error probability is at least $\frac{1}{L+1}$.

Given the above intuition, we proceed with the formal analysis as follows. Let $T$ denote the outcome of James' coin flip. 
The average error probability is at least
\begin{align}
&\frac{1}{2}\cdot\prob{\left.\forall\cL\in\binom{\curbrkt{0,1,\cdots,L}}{L},\;\normtwo{\sum_{i\in\cL}\vbfx_i - L\vu}\le\sqrt{nN};\;\vbfx_0,\vbfx_1,\cdots,\vbfx_L\text{ are distinct}\right|T=1}\cdot\frac{1}{L+1}\label{eqn:gavc_converse_case1}\\
+&\frac{1}{2}\cdot\prob{\left.\forall\cL\in\binom{\curbrkt{0,1,\cdots,L}}{L},\;\normtwo{\sum_{i\in\cL}\vbfy_i - L\vv}\le\sqrt{nN};\;\vbfy_0,\vbfy_1,\cdots,\vbfy_L\text{ are distinct}\right|T=0}\cdot\frac{1}{L+1}.\label{eqn:gavc_converse_case2}
\end{align}
By symmetry of the cases where $T=1$ and $T=0$, it suffices to bound term \eqref{eqn:gavc_converse_case1}. Note that conditioning on $T$ can be removed since the events in the probability are independent of $T$.
\begin{align}
&\prob{\forall\cL\in\binom{\curbrkt{0,1,\cdots,L}}{L},\;\normtwo{\sum_{i\in\cL}\vbfx_i - L\vu}\le\sqrt{nN};\;\vbfx_0,\vbfx_1,\cdots,\vbfx_L\text{ are distinct}}\notag\\
\ge&\prob{\forall\cL\in\binom{\curbrkt{0,1,\cdots,L}}{L},\;\normtwo{\sum_{i\in\cL}\vbfx_i - L\vu}\le\sqrt{nN}} - \prob{\vbfx_0,\vbfx_1,\cdots,\vbfx_L\text{ are \emph{not} distinct}}.\label{eqn:gavc_converse_two_terms}
\end{align}
The second term in Eqn. \eqref{eqn:gavc_converse_two_terms} equals
\begin{align}
1-\frac{M-1}{M}\frac{M-2}{M}\cdots\frac{M-L}{M} \eqcolon& \delta_n .\notag
\end{align}
Note that $\delta_n = o_n(1) $.

The first term in Eqn. \eqref{eqn:gavc_converse_two_terms} is more involved. 
Let $\cB\coloneqq\cB^n\paren{\vu,\sqrt{nP_1\eta}}$ denote the $n$-dimensional Euclidean ball centered around $\vu$ of radius $\sqrt{nP_1\eta}$.
Note that by Eqn. \eqref{eqn:prob_x_in_k}, $\prob{\vbfx\in\cB}\ge{\eps/2}$.
We claim that the first term can be bounded as follows.
\begin{align}
&\prob{\bigcap_{\cL\in\binom{\curbrkt{0,1,\cdots,L}}{L}}\curbrkt{\normtwo{\sum_{i\in\cL}(\vbfx_i-\vu)}^2\le nN}}\notag\\
\ge&\prob{ \bigcap_{i\in\curbrkt{0,1,\cdots,L}}\curbrkt{\vbfx_i\in\cB} \cap \bigcap_{\substack{i\ne j\\i,j\in\curbrkt{0,1,\cdots,L}}}\curbrkt{\inprod{\vbfx_i - \vu}{\vbfx_j-\vu}\le nP_1\sqrt{\gamma\eta}} }.\label{eqn:implies_bound_on_normsquared}
\end{align}
This follows due to the following  reasons. Assume the event in Eqn. \eqref{eqn:implies_bound_on_normsquared} holds. Then 
for any $\cL\in\binom{\curbrkt{0,1,\cdots,L}}{L}$, we have{}
\begin{align}
\normtwo{\sum_{i\in\cL}(\vbfx_i-\vu)}^2 = &\sum_{i\in\cL}\normtwo{\vbfx_i-\vu}^2+\sum_{\substack{i\ne j\\i,j\in\cL}}\inprod{\vbfx_i - \vu}{\vbfx_j - \vu}\notag\\
\le& LnP_1\eta + L(L-1)nP_1\sqrt{\gamma\eta}\notag\\
=&nLP_1\paren{ \eta_*+\gamma/2 + (L-1)\sqrt{\gamma\eta} }\label{eqn:by_choice_of_eta}\\
\le&nLP_1\paren{1 + \gamma/2 + (L-1)\sqrt{\gamma\eta}},\label{eqn:bound_on_etastar}\\
\le&nLP_1\paren{1+\delta_1/2}\label{eqn:by_choice_of_various_param}\\
<&nN,\label{eqn:by_gap_btwn_lp_n}
\end{align}
where Eqn. \eqref{eqn:by_choice_of_eta} is by definition of $\eta\coloneqq\eta_*+\gamma/2$, Eqn. \eqref{eqn:bound_on_etastar} follows since $\eta_*\le1$, Eqn. \eqref{eqn:by_choice_of_various_param} is by the choice of $\gamma$ (Eqn. \eqref{eqn:def_gamma}) such that $ \gamma/2+(L-1)\sqrt{\gamma\eta} \le \delta_1/2 $, and Eqn. \eqref{eqn:by_gap_btwn_lp_n} follows since $LP_1(1+\delta_1) = N$.

Continuing with Eqn. \eqref{eqn:implies_bound_on_normsquared}, it can be further lower bounded by
\begin{align}
\prob{ \bigcap_{i\in\curbrkt{0,1,\cdots,L}}\curbrkt{\vbfx_i\in\cB} } - \prob{ \bigcap_{i\in\curbrkt{0,1,\cdots,L}}\curbrkt{\vbfx_i\in\cB} \cap \bigcup_{\substack{i\ne j\\i,j\in\curbrkt{0,1,\cdots,L}}}\curbrkt{\inprod{\vbfx_i-\vu}{\vbfx_j-\vu}>nP_1\sqrt{\gamma\eta}} }.\label{eqn:another_two_terms_to_bd}
\end{align}
The first term of Eqn. \eqref{eqn:another_two_terms_to_bd} equals
\begin{align}
\prob{\vbfx\in\cB}^{L+1}\ge&({\eps/2})^{L+1}.
\end{align}
since each $\vbfx_i$ is i.i.d. The second term of Eqn. \eqref{eqn:another_two_terms_to_bd} is upper bounded as 
\begin{align}
& \prob{ \bigcap_{i\in\curbrkt{0,1,\cdots,L}}\curbrkt{\vbfx_i\in\cB} \cap \bigcup_{\substack{i\ne j\\i,j\in\curbrkt{0,1,\cdots,L}}}\curbrkt{\inprod{\vbfx_i-\vu}{\vbfx_j-\vu}>nP_1\sqrt{\gamma\eta}} }\notag\\
=&\prob{    \bigcup_{\substack{i\ne j\\i,j\in\curbrkt{0,1,\cdots,L}}} \paren{ \curbrkt{\inprod{\vbfx_i-\vu}{\vbfx_j-\vu}>nP_1\sqrt{\gamma\eta}} \cap  \bigcap_{i\in\curbrkt{0,1,\cdots,L}}\curbrkt{\vbfx_i\in\cB} } }\notag\\
\le&\frac{(L+1)L}{2}\prob{ \curbrkt{\inprod{\vbfx'-\vu}{\vbfx''-\vu}>nP_1\sqrt{\gamma\eta}}\cap\curbrkt{\vbfx'\in\cB}\cap\curbrkt{\vbfx''\in\cB} }.\label{eqn:bound_ip}
\end{align}
To upper bound the probability in Eqn. \eqref{eqn:bound_ip}, note that the event in the probability implies
\begin{align}
\normtwo{(\vbfx'-\vu) - \sqrt{\gamma/\eta}(\vbfx''-\vu)}^2 = &\normtwo{\vbfx'-\vu}^2 + \frac{\gamma}{\eta}\normtwo{\vbfx''-\vu}^2 - 2\sqrt{\gamma/\eta}\inprod{\vbfx'-\vu}{\vbfx''-\vu}\notag\\
<&nP_1\eta + \frac{\gamma}{\eta}nP_1\eta - 2\sqrt{\gamma/\eta}nP_1\sqrt{\gamma\eta}\notag\\
=&nP_1(\eta-\gamma)\notag\\
=&nP_1(\eta_*-\gamma/2).\notag
\end{align}
Therefore, we have
\begin{align}
&\prob{ \curbrkt{\inprod{\vbfx'-\vu}{\vbfx''-\vu}>nP_1\sqrt{\gamma\eta}}\cap\curbrkt{\vbfx'\in\cB}\cap\curbrkt{\vbfx''\in\cB} }\notag\\
\le&\prob{\normtwo{\vbfx' - \paren{\vu+\sqrt{\gamma/\eta}(\vbfx''-\vu)}}^2<nP_1(\eta_*-\gamma/2)}.\notag
\end{align}
By the definition of $\eta_*$, 
\begin{align}
\liminf_{n\to\infty}\sup_{\vu'\in\bR^n}\prob{ \normtwo{\vbfx' - \vu'}^2<nP_1(\eta_*-\gamma/2)  }= 0.\notag
\end{align}
In other words,  take $\vu'\coloneqq\vu+\sqrt{\gamma/\eta}(\vbfx''-\vu)$, then 
\begin{align}
\prob{ \normtwo{\vbfx'- \paren{ \vu+\sqrt{\gamma/\eta}(\vbfx''-\vu) }}^2<nP_1(\eta_*-\gamma/2) } = \delta_n' ,\notag
\end{align}
where $ \delta_n' = o_n(1) $.
Substituting it back to Eqn. \eqref{eqn:bound_ip}, the second term in Eqn. \eqref{eqn:another_two_terms_to_bd} is at most $\frac{(L+1)L}{2}\delta_n' $.

Finally, for $T=1$ case,  the average error probability (Eqn. \eqref{eqn:gavc_converse_case1}) is at least
\begin{align}
\frac{1}{2(L+1)}\paren{({\eps/2})^{L+1} - \frac{(L+1)L}{2}\delta_n' -\delta_n }.\notag
\end{align}
By similar calculations, the average probability of error when $T=0$ is also bounded away from 0. This finishes the proof for converse.

\section{Open questions and future directions}\label{sec:open}
We list several open questions and future directions.
\begin{itemize}
  \item In an ongoing work \cite{zhang-jaggi-2020-listdec_obli_avc_inprep}, Zhang and Jaggi managed to close the gap between upper and lower bounds on list sizes for list decoding for oblivious AVCs under input and state constraints. 
  This is achieved by introducing yet another new notion of symmetrizability named \emph{$ \cp $-symmetrizability} (where $\cp$ stands for \emph{completely positive}). It is believed that in the oblivious case $ \cp $-symmetrizability collapses to weak symmetrizability introduced by \cite{sarwate-gastpar-2012-list-dec-avc-state-constr}, which, if is true, will prove a conjecture left in \cite{sarwate-gastpar-2012-list-dec-avc-state-constr}. 
  This will be justified in a future version of \cite{zhang-jaggi-2020-listdec_obli_avc_inprep}. 
  In the AVMAC setting, it is natural to import ideas from \cite{zhang-jaggi-2020-listdec_obli_avc_inprep} and check how $ \cp $-symmetrizability should be defined properly and what it yields.
  This is left as one of our future directions.
  \item In a recent work \cite{pereg-steinberg-2019-obli-avmac-with-wo-constr} which dealt with \emph{unique} decoding for two-user AVMACs, the boundary case where exactly one user has capacity 0 was solved which was left as an open question in \cite{ahlswedecai-1999-obli-avmac-no-constr}. 
  This does not directly follow from single-user symmetrizability since the  user who transmits at zero-rate may use \emph{nonempty} codebook of subexponential size. 
  This increases the difficulty for James to jam.
  The boundary case for list decoding for AVMACs will be treated in a future version of this paper. 
\end{itemize}

\section{Acknowledgement}
YZ would like to thank Sidharth Jaggi for helpful discussions and encouragement. 

\bibliographystyle{alpha}
\bibliography{IEEEabrv,ref} 

\begin{thebibliography}{DMOZ19}

\bibitem[AC99]{ahlswedecai-1999-obli-avmac-no-constr}
Rudolf Ahlswede and Ning Cai.
\newblock Arbitrarily varying multiple-access channels. i. ericson's
  symmetrizability is adequate, gubner's conjecture is true.
\newblock {\em IEEE Transactions on Information Theory}, 45(2):742--749, 1999.

\bibitem[BBT60]{blackwell-avc-1960}
David Blackwell, Leo Breiman, and A.~J. Thomasian.
\newblock {The Capacity of a Class of Channels under Random Coding}.
\newblock {\em Ann. of Mathematical Statistics}, 31(3):558--567, 1960.

\bibitem[Cai16]{cai-2016-list-dec-obli-avmac}
Ning Cai.
\newblock List decoding for arbitrarily varying multiple access channel
  revisited: List configuration and symmetrizability.
\newblock {\em IEEE Transactions on Information Theory}, 62(11):6095--6110,
  2016.

\bibitem[CJ81]{csiszar-korner-1981}
I.~Csisz{\'a}r and J.K{\"o}rner.
\newblock {On the capacity of the arbitrarily varying channel for maximum
  probability of error}.
\newblock {\em Z. Wahrscheinlichkeitstheorie Verv. Gebiete}, 57:87--101, 1981.

\bibitem[CJM15]{chen-et-al}
Z.~Chen, S.~Jaggi, and M.Langberg.
\newblock {A Characterization of the Capacity of Online (causal) Binary
  Channels}.
\newblock In {\em {Proc. {ACM} Symp. on Discrete Algorithms (SODA)}}, Portland,
  U.S.A, June 2015.

\bibitem[CN88]{csiszar-narayan-it1988-obliviousavc}
Imre Csisz{\'a}r and Prakash Narayan.
\newblock {The Capacity of the Arbitrarily Varying Channel Revisited :
  Positivity, Constraints}.
\newblock {\em {IEEE} Trans. Inf. Theory}, 34:181--193, 1988.

\bibitem[CN91]{csiszar-narayan-1991-gavc}
Imre Csisz{\'a}r and Prakash Narayan.
\newblock Capacity of the gaussian arbitrarily varying channel.
\newblock {\em IEEE Transactions on Information Theory}, 37(1):18--26, 1991.

\bibitem[DKS18]{dks-2018-list_dec_est_and_learn}
Ilias Diakonikolas, Daniel~M Kane, and Alistair Stewart.
\newblock List-decodable robust mean estimation and learning mixtures of
  spherical gaussians.
\newblock In {\em Proceedings of the 50th Annual ACM SIGACT Symposium on Theory
  of Computing}, pages 1047--1060. ACM, 2018.

\bibitem[DMOZ19]{dmoz-2019-pseudorandom_from_hardness}
Dean Doron, Dana Moshkovitz, Justin Oh, and David Zuckerman.
\newblock Nearly optimal pseudorandomness from hardness.
\newblock Technical report, ECCC preprint TR19-099, 2019.

\bibitem[Eli57]{elias-1957-listdec}
Peter Elias.
\newblock List decoding for noisy channels.
\newblock 1957.

\bibitem[GL89]{goldreich-levin-1989-hardcorepredicate-listdec}
Oded Goldreich and Leonid~A Levin.
\newblock A hard-core predicate for all one-way functions.
\newblock In {\em Proceedings of the twenty-first annual ACM symposium on
  Theory of computing}, pages 25--32. ACM, 1989.

\bibitem[Gur06]{guruswami-2006-list_dec_avgcasecplx}
Venkatesan Guruswami.
\newblock List decoding in average-case complexity and pseudorandomness.
\newblock In {\em 2006 IEEE Information Theory Workshop-ITW'06 Punta del Este},
  pages 32--36. IEEE, 2006.

\bibitem[HK19]{hosseinigoki-kosut-2018-oblivious-gaussian-avc-ld}
Fatemeh Hosseinigoki and Oliver Kosut.
\newblock List-decoding capacity of the gaussian arbitrarily-varying channel.
\newblock {\em Entropy}, 21(6):575, 2019.

\bibitem[Hug97]{hughes-1997-list-avc}
Brian~L. Hughes.
\newblock The smallest list for the arbitrarily varying channel.
\newblock {\em IEEE Transactions on Information Theory}, 43(3):803--815, 1997.

\bibitem[PS19]{pereg-steinberg-2019-obli-avmac-with-wo-constr}
Uzi Pereg and Yossef Steinberg.
\newblock The capacity region of the arbitrarily varying mac: With and without
  constraints.
\newblock {\em arXiv preprint arXiv:1901.00939}, 2019.

\bibitem[SG12]{sarwate-gastpar-2012-list-dec-avc-state-constr}
Anand~D Sarwate and Michael Gastpar.
\newblock List-decoding for the arbitrarily varying channel under state
  constraints.
\newblock {\em IEEE Transactions on Information Theory}, 58(3):1372--1384,
  2012.

\bibitem[Woz58]{wozencraft-1958-listdec}
John~M Wozencraft.
\newblock List decoding.
\newblock {\em Quarterly Progress Report}, 48:90--95, 1958.

\bibitem[ZJ]{zhang-jaggi-2020-listdec_obli_avc_inprep}
Yihan Zhang and Sidharth Jaggi.
\newblock {List decoding for constrained oblivious {AVCs}: closing the gap}.
\newblock In preparation.

\end{thebibliography}

\end{document}